\documentclass[11pt]{article}
\usepackage[latin9]{inputenc}
\usepackage{amsmath}
\usepackage{amssymb}

\makeatletter
\usepackage{amstext}
\usepackage{amsfonts}
\usepackage{amsthm}
\usepackage{bbold}
\usepackage{mathtools}
\usepackage{scrextend}
\usepackage{array}
\usepackage{multicol}
\usepackage{mathrsfs}
\usepackage{color}
\usepackage{hyperref}
\usepackage{bm}
\usepackage{ textcomp }
\usepackage{mleftright}
\usepackage{IEEEtrantools}
\usepackage[colorinlistoftodos]{todonotes}
\usepackage{pgfplots}
\pgfplotsset{compat=1.15}
\usetikzlibrary{arrows}

\newtheorem{theorem}{Theorem}[section]    
\newtheorem{claim}[theorem]{Claim}    
\newtheorem{corollary}[theorem]{Corollary}    
\newtheorem{proposition}[theorem]{Proposition}    
\newtheorem{lemma}[theorem]{Lemma}    

\renewcommand{\qed}{\hfill{$\rule{6pt}{6pt}$}} 
\renewenvironment{proof}{\noindent{\bf Proof:}}{\qed\\}
\newenvironment{proofof}[1]{\noindent{\bf Proof of #1:}}{\qed\\}
\newenvironment{remark}{\noindent{\bf Remark:}}

\numberwithin{equation}{section}

\newcommand{\complex}{{\mathbb C}}

\newcommand{\integers}{{\mathbb Z}}

\newcommand{\ket}[1]{\left| #1 \right\rangle}

\newcommand{\ketbra}[2]{| #1 \rangle\!\langle #2 |}

\newcommand{\trnorm}[1]{\left\| #1 \right\|_{\mathrm{1}}}
\newcommand{\size}[1]{\left| #1 \right|}
\newcommand{\set}[1]{\left\{ #1 \right\}}
\newcommand{\floor}[1]{{\lfloor #1 \rfloor}}
\newcommand{\ceil}[1]{{\left \lceil #1 \right \rceil}}
\newcommand{\density}[1]{\ketbra{#1}{#1}}
\newcommand{\eqdef}{\coloneqq}

\newcommand{\complexi}{{\mathrm{i}}}

\newcommand{\abs}[1]{\left| #1 \right|}
\newcommand{\tpmod}[1]{{\@displayfalse\pmod{#1}}}
\newcommand{\br}[1]{\mleft(#1\mright)}

\newcommand*\diff{\mathop{}\!\mathrm{d}}

\DeclareMathOperator{\trace}{Tr}
\DeclareMathOperator{\support}{supp}
\DeclareMathOperator{\Order}{{\mathrm{O}}}

\newcommand{\id}{{\mathbb 1}}

\newcommand{\linear}{{\mathsf L}}
\newcommand{\unitary}{{\mathsf U}}
\newcommand{\qstate}{{\mathsf D}}

\newcommand{\tensor}{\otimes}

\newcommand{\suppress}[1]{}
\newcommand{\comment}[1]{}

\DeclareMathOperator{\rS}{S}
\DeclareMathOperator{\rI}{I}
\DeclareMathOperator{\rF}{F}
\DeclareMathOperator{\rD}{D}

\DeclareMathOperator{\rP}{P}
\DeclareMathOperator{\rV}{V}

\DeclareMathOperator{\Normal}{\Phi}

\newcommand{\ri}{{\mathrm i}}
\newcommand{\rii}{{\mathrm{ii}}}

\DeclareMathOperator{\markdist}{\Delta}
\DeclareMathOperator{\maxmarkdist}{\Delta_{\max}}

\newcommand{\Dh}{{\mathrm D}_{\mathrm H}}

\newcommand{\Imax}{{\mathrm I}_{\max}}

\newcommand{\Dmax}{{\mathrm D}_{\max}}

\newcommand{\Pos}{{\mathsf{Pos}}}

\newcommand{\sB}{{\mathsf B}}

\newcommand{\ME}{{\mathsf {ME}}_{R - B - C}}
\newcommand{\MC}{{\mathsf {MC}}_{R - B - C}}
\newcommand{\QMC}{{\mathsf {QMC}}_{R - B - C}}

\newcommand{\cD}{{\mathcal D}}
\newcommand{\cE}{{\mathcal E}}
\newcommand{\cH}{{\mathcal H}}

\newcommand{\cK}{{\mathcal K}}

\newcommand{\Pitilde}{{\widetilde \Pi}}
\newcommand{\ttau}{{\widetilde \tau}}
\newcommand{\tsigma}{{\widetilde \sigma}}

\newcommand{\barrho}{{\overline \rho}}

\newcommand {\email} [1] {\href{mailto:#1}{\texttt{#1}}}

\makeatother

\begin{document}

\definecolor{yqqqqq}{rgb}{0.5019607843137255,0.,0.}

\title{
\textbf{One-shot quantum state redistribution and quantum Markov chains}\footnote{This paper was presented at the sixteenth conference on the Theory of Quantum Computation, Communication, and Cryptography (TQC 2021) and 2021 IEEE International Symposium on Information Theory (ISIT 2021). An extended abstract of it is published in the proceedings of ISIT 2021\cite{ABJNT21-QSR}. This work was also part of S.B.'s PhD thesis. }
}

\date{}

\author{
Anurag Anshu~\thanks{School of Engineering and Applied Sciences, Harvard University, Cambridge, MA. Email: \email{anuraganshu@seas.harvard.edu}~. Part of the work was done when the author was affiliated to the Department of Combinatorics and Optimization $\&$ the Institute for Quantum Computing, University of Waterloo and the Perimeter Institute for Theoretical Physics, Waterloo, Canada.} \\
Harvard University
\and
Shima Bab Hadiashar~\thanks{Department of Combinatorics and Optimization,
and Institute for Quantum Computing, University
of Waterloo, 200 University Ave.\ W., Waterloo, ON,
N2L~3G1, Canada. Email: \texttt{\{sbabhadi,ashwin.nayak\}@uwaterloo.ca}~.
Research supported in part by NSERC Canada. SBH is also supported 
by an Ontario Graduate Scholarship.} \\
U.\ Waterloo
\and
Rahul Jain~\thanks{Department of Computer Science, and Center for Quantum Technologies, National University of Singapore, 21 Lower Kent Ridge Rd, Singapore 119077. Email: \email{rahul@comp.nus.edu.sg}~.} \\
National U.\ Singapore
\and
Ashwin Nayak~\footnotemark[2] \\
U.\ Waterloo
\and
Dave Touchette~\thanks{Department of Computer Science,
and Institut Quantique, Universit\'e
de Sherbrooke, 2500 Boulevard de l'Universit\'e, Sherbrooke, QC J1K 2R1, Canada. Email: \email{dave.touchette@usherbrooke.ca}~.}\\
U.\ Sherbrooke
}

\maketitle

\begin{abstract}
 We revisit the task of quantum state redistribution in the one-shot setting, and design a protocol for this task with communication cost in terms of a measure of distance from quantum Markov chains. More precisely, the distance is defined in terms of quantum max-relative entropy and quantum hypothesis testing entropy. 
 
 Our result is the first to operationally connect quantum state redistribution and quantum Markov chains, and can be interpreted as an operational interpretation for a possible one-shot analogue of quantum conditional mutual information. The communication cost of our protocol is lower than all previously known ones and asymptotically achieves the well-known rate of quantum conditional mutual information. Thus, our work takes a step towards an optimal characterization of the resources required for one-shot quantum state redistribution, an important open problem in quantum Shannon theory.
\end{abstract}

\section{Introduction}
\label{sec:QSR-Intro}

\subsection{Background and result}

The connection between conditional mutual information and Markov chains has led to a rich body of results in classical computer science and information theory. It is well known that for any tripartite distribution $P^{RBC}$ over registers~$RBC$, the conditional mutual information
$$ \rI(R:C\,|\,B)_P \quad = \quad \min_{Q^{RBC} \;\in\;  \MC} \rD \!\big( P^{RBC} \big\| Q^{RBC} \big) \enspace,$$
where $\MC$ is the set of Markov distributions~$Q$, i.e., those that satisfy $\rI(R:C\,|\,B)_Q = 0$, and~$\rD(\cdot \| \cdot)$ is the relative entropy function. In fact, one can choose a distribution~$Q$ achieving the minimum above with $Q^{RB}=P^{RB}$, $Q^{BC}= P^{BC}$. In the quantum case, the above identity fails drastically. For an example presented in ref.~\cite{CSW12-antisymmetric-state-entanglement} (see also ref.~\cite[Section VI]{Ibinson2008}), the right-hand side is a constant, whereas the left-hand side approaches zero as the system size increases. Given this, it is natural to ask if there is an extension of the classical identity to the quantum case. This has been shown to be true in  a sense that for any tripartite quantum state $\psi^{RBC}$, it holds that
\begin{equation}
\label{eq:condmutrelent}
\rI(R:C \,|\, B)_{\psi} \quad = \quad \min_{\sigma^{RBC}\in \QMC}\br{\rD\br{\psi^{RBC}\|\sigma^{RBC}}-\rD\br{\psi^{BC}\|\sigma^{BC}}} \enspace,
\end{equation}
where $\QMC$ is the set of quantum states $\sigma$ satisfying $\rI(R:C \,|\, B)_{\sigma}=0$, $\psi^{RB}=\sigma^{RB}$~\cite{BertaSW15}.
(For completeness, we provide a proof in Section \ref{sec:QMarkov}, Lemma~\ref{lem:condmutrelent}.) The difference between the quantum and the classical expressions can now be understood as follows. For the classical case, the closest Markov chain $Q$ to a distribution $P$ (in relative entropy) satisfies the aforementioned relations~$Q^{RB}=P^{RB}$ and~$Q^{BC}=P^{BC}$. Thus, the second relative entropy term in Eq.~(\ref{eq:condmutrelent}) vanishes. In the quantum case, due to monogamy of entanglement we cannot in general ensure that $\sigma^{BC}=\psi^{BC}$. Thus, the quantum relative entropy distance to quantum Markov chains can be bounded away from the quantum conditional mutual information.

In this work, we prove a one-shot analogue of Eq.~\eqref{eq:condmutrelent}. This is achieved in an operational manner, by showing that a one-shot analogue of the right-hand side in Eq.~(\ref{eq:condmutrelent}) is the achievable communication cost of the \emph{quantum state redistribution\/} of~$\ket{\psi}^{RABC}$, a purification of~$\psi^{RBC}$. In the task of quantum state redistribution, the pure quantum state~$\ket{\psi}^{RABC}$ is known to two parties, Alice and Bob, and is shared between Alice (who has registers~$AC$), Bob (who has~$B$), and a reference party, Ref (who has~$R$). Additionally, Alice and Bob may share an arbitrary pure entangled state. The goal is to transmit the content of register~$C$ to Bob using a communication protocol involving only Alice and Bob, in such a way that all correlations, including those with Ref, are approximately preserved. (See Figure~\ref{fig:QSR} for an illustration of state redistribution.) Given a quantum state~$\phi^{RBC}$, we identify a natural subset of Markov extensions of~$\phi^{RB}$, which we denote by~$\ME^{\epsilon,\phi}$ and define formally at the end of Section~\ref{sec:QMarkov}, in Eq.~\ref{eq:ME-def}. We establish the following result in terms of the max-relative entropy~($\Dmax$) and~$\epsilon$-hypothesis testing relative entropy~($\Dh^\epsilon$) functions.

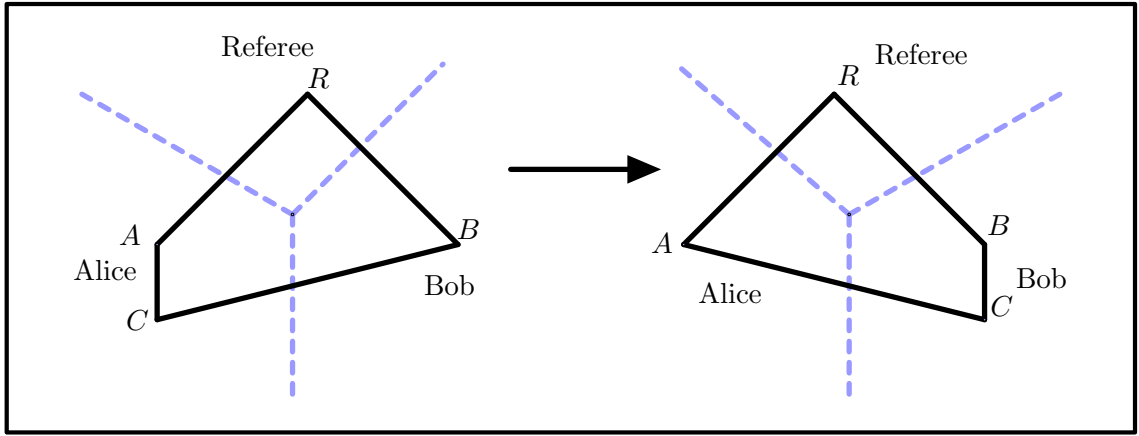
\begin{figure}[htbt]
\centering
\begin{tikzpicture}[ line cap=round,line join=round,>=triangle 45,x=1.0cm,y=1.0cm]

\draw[ ultra thick, fill= blue!10!white] (-4.7, 0) rectangle (9.7,5);

\draw [->,line width=1pt] (1.7,3.) -- (3.7,3.);

\draw [thick, fill=blue!30!white] (-1.,4.) -- (-3.,2.) -- (-3.1,1.) -- (1.,2.) -- (-1.,4.);


\draw [ thick, fill=blue!30!white]  (6.,4.) -- (4.,2.) -- (8.,1.) -- (8.,2.) -- (6.,4.);


\draw [line width=1pt,dash pattern=on 4pt off 4pt,color= orange!80!red] (-1.2,2.4)-- (0.8,4.4);
\draw [line width=1pt,dash pattern=on 4pt off 4pt,color=orange!80!red] (-1.2,2.4)-- (-4.,4.);
\draw [line width=1pt,dash pattern=on 4pt off 4pt,color=orange!80!red] (-1.2,2.4)-- (-1.2,0.8);
\draw[color=black] (-1.2,0.5) node {$\ket{\psi}^{RABC}$};

\draw [line width=1pt,dash pattern=on 4pt off 4pt,color=orange!80!red] (6.2,2.4)-- (6.2,0.8);
\draw [line width=1pt,dash pattern=on 4pt off 4pt,color=orange!80!red] (6.2,2.4)-- (9.,4.);
\draw [line width=1pt,dash pattern=on 4pt off 4pt,color=orange!80!red] (6.2,2.4)-- (3.9,4.4);
\draw[color=black] (6.2,0.5) node {$\phi^{RABC} \approx^{\epsilon} \ket{\psi}^{RABC}$};

\draw[color=black] (-0.86,4.21) node {$R$};
\draw[color=black] (-3.36,2.13) node {$A$};
\draw[color=black] (-3.4,0.99) node {$C$};
\draw[color=black] (1.34,2.21) node {$B$};
\draw[color=black] (6.19,4.26) node {$R$};
\draw[color=black] (3.71,2.02) node {$A$};
\draw[color=black] (8.23,1.18) node {$C$};
\draw[color=black] (8.19,2.26) node {$B$};
\draw[color=orange!50!red] (-2.5,4.4) node {The referee};
\draw[color=orange!50!red] (-4,1.66) node {Alice};
\draw[color=orange!50!red] (0.89,1.2) node { Bob};
\draw[color=orange!50!red] (4.62,1.35) node {Alice};
\draw[color=orange!50!red] (9,1.55) node {Bob};
\draw[color=orange!50!red] (7.86,4.4) node {The referee};
\end{tikzpicture}

   \caption{An illustration of quantum state redistribution.}
   \label{fig:QSR}
\end{figure}

\begin{theorem}
\label{thm:oneshotcondmutrelent}
For any~$\epsilon \in (0, 1/100)$ and pure quantum state $\ket{\psi}^{RABC}$, the quantum communication cost of redistributing the register $C$ from Alice (who initially holds $AC$) to Bob (who initially holds $B$) with error~$10 \sqrt{\epsilon}$ is at most
$$\frac{1}{2} \min_{\psi' \in \sB^{\epsilon}(\psi^{RBC})} \min_{\sigma^{RBC}\in \ME^{\epsilon^2/4,\psi'}} \mleft[ \Dmax \br{\psi'^{RBC} \big\| \sigma^{RBC}}-\Dh^{\epsilon}\br{\psi'^{BC} \big\|\sigma^{BC}} \mright]+ \Order \!\left( \log\frac{1}{\epsilon} \right) \enspace.$$
\end{theorem}
The difference between minimizing over the set~$\ME^{\epsilon^2/4,\psi'}$ versus~$\QMC$ is best understood from the definitions in Section~\ref{sec-background};we give a brief explanation of the difference and why the set~$\ME^{\epsilon^2/4,\psi'}$ is considered in Section~\ref{sec-intro-techniques}. 
We believe the above result can be stated in terms of a minimization over all of~$\QMC$. 
In the above bound, there is an additional minimization over the set~$\sB^{\epsilon}(\psi^{RBC})$, which is an~$\epsilon$-neighbourhood of~$\psi$ (see Section~\ref{sec-background} for a formal definition). Considering~$\epsilon$ perturbations of the state in question may result in significantly lower communication, at the cost of increasing the error in the output state by at most~$\epsilon$. This also allows us to achieve the optimal rate in the asymptotic i.i.d.\ setting. The information-theoretic quantities appearing in the above bound arise from two subroutines on which the underlying protocol is based --- Coherent Rejection Sampling (building on the Convex-Split Lemma) and Position-Based Decoding. Smooth max-relative entropy and smooth hypothesis testing relative entropy, respectively, are precisely the quantities which appear in the analysis of these subroutines.

The protocol that achieves the bound in Theorem~\ref{thm:oneshotcondmutrelent} is reversible. So, in order to redistribute~$C$ 
 from Alice to Bob, Alice and Bob can instead run the time-reversal of 
 the protocol in which register~$C$ is initially with Bob and he wants to send it to Alice. 
 This implies the following corollary.
 \begin{corollary}
\label{cor:cor-main}
For any pure quantum state $\ket{\psi}^{RABC}$, the quantum communication cost of redistributing the register $C$ from Alice (who initially holds $AC$) to Bob (who initially holds $B$) with error~$10 \sqrt{\epsilon}$ is at most the minimum of
\[
\frac{1}{2} \inf_{\psi' \in \sB^{\epsilon}(\psi^{RBC})} \inf_{\sigma^{RBC}\in \ME^{\epsilon^2/4,\psi'}} \mleft[ \Dmax \br{\psi'^{RBC}\|\sigma^{RBC}}-\Dh^{\epsilon}\br{\psi'^{BC}\|\sigma^{BC}} \mright]+  \Order \!\left( \log\frac{1}{\epsilon} \right) 
\]
and
\[
\frac{1}{2} \inf_{\psi' \in \sB^{\epsilon}(\psi^{RAC})} \inf_{\sigma^{RAC}\in \mathsf{ME}_{R - A - C}^{\epsilon^2/4,\psi'}} \mleft[ \Dmax \br{\psi'^{RAC}\|\sigma^{RAC}}-\Dh^{\epsilon}\br{\psi'^{AC}\|\sigma^{AC}} \mright]+  \Order \!\left( \log\frac{1}{\epsilon} \right) \enspace.
\]
\end{corollary}
Connections between quantum Markov chains and special cases of quantum state redistribution have been made, possibly implicitly, in several previous works. An example is in the compression of mixed states; see, e.g.,~\cite[Section VIII.E]{KN01-compression}. However, as far as we know, Theorem~\ref{thm:oneshotcondmutrelent} is the first result that operationally connects the \emph{cost\/} of quantum state redistribution in its most general form to a measure of distance from quantum Markov chains (even in the asymptotic i.i.d.\ setting). The best previously known achievable one-shot bound for the communication cost of state redistribution, namely,
\begin{equation}
 \label{eq:AJW-bound}
     \frac{1}{2} \inf_{\sigma^{C}} \inf_{\psi' \in \sB^{\epsilon}(\psi^{RBC})} \left(  \Dmax \mleft( \psi'^{RBC} \| \psi'^{RB} \tensor \sigma^{C} \mright) - \Dh^{\epsilon^{2}} \mleft( \psi'^{BC} \| \psi'^{B} \tensor \sigma^{C} \mright)
	 \right) + \log \frac{1}{\epsilon^2} \enspace,
\end{equation}
when the state~$\ket{\psi}^{RABC}$ is redistributed with error~$\Order(\epsilon)$ was due to 
Anshu, Jain, and Warsi~\cite{AJW18-Achievability-QSR}. Note that~$\sigma^{C} \eqdef \psi'^{C}$
is a nearly optimal solution for Eq.~\eqref{eq:AJW-bound} as discussed in 
ref.~\cite{CBR14-smooth-Imax}, and the product state~$\psi'^{RB} \tensor \psi'^{C}$ is a 
Markov state in the set~$\ME^{\epsilon^2/4,\psi'}$. So, the bound in 
Theorem~\ref{thm:oneshotcondmutrelent} is smaller than that in Eq.~\eqref{eq:AJW-bound} in the sense that the 
minimization is over a larger set. In the special case 
where~$\psi^{RBC}$ is a quantum Markov chain, our protocol has near-zero communication. This feature is not present in other protocols and their communication may be as large as~$(1/2) \log |C|$. 
Moreover, in the case that register~$A$, or~$B$,  or both~$A$ and~$B$ are trivial, our bound reduces to $\frac{1}{2}\Imax^{\epsilon}(R:C)$. The three cases correspond to state splitting, state merging, and compression without side-information, respectively, for which this bound is known to be the optimal communication cost in the one-shot case.
\suppress{
In 
particular, if register~$R$ is trivial, our protocol achieves the optimal cost of zero qubits
of communication by choosing~$\sigma^{BC} \eqdef \psi'^{BC}$, whereas the cost of the 
protocol in ref.~\cite{AJW18-Achievability-QSR}, stated in Eq.~\eqref{eq:AJW-bound} may be as 
large as~$(1/2) \log |C|$.
}

\subsection{Techniques}
\label{sec-intro-techniques}

The protocol we design is most easily understood by considering a folklore protocol for redistributing quantum Markov states.
In the case that~$\psi^{RBC}$ is a Markov state, its purification~$\ket{\psi}^{RABC}$ can be transformed through local isometry operators~$V_1 : A \rightarrow A^R J' A^C$ and~$V_2 : B \rightarrow B^R J B^C$ into the following:
\begin{equation}
\label{eq:QMC-extension}
    ( V_1 \tensor V_2 ) \ket{\psi}^{RABC} \quad = \quad \sum_j \sqrt{p(j)} \; \ket{\psi_j}^{RA^RB^R} \tensor \ket{jj}^{JJ'} \tensor \ket{\psi_j}^{A^C B^C C} \enspace.
\end{equation}
The existence of isometries~$V_1$ and~$V_2$ is a consequence of the special structure of quantum Markov states proved by Hayden, Josza, Petz, and Winter~\cite{HJPW04-Qntm-Markov-state}.
Note that after the above transformation, conditioned on registers~$J$ and~$J'$, systems~$R A^R B^R$ are decoupled from systems~$A^C C B^C$. 
So using the embezzling technique due to van Dam and Hayden~\cite{vDH03-embezzling}, conditioned on~$J$ and~$J'$, Alice and Bob can first embezzle-out systems~$A^C C B^C$ and then embezzle-in the same systems \emph{but now with system~$C$ on Bob's side\/} such that at the end the global state is close to the state in Eq.~\eqref{eq:QMC-extension}. 
This protocol incurs no communication; see Fig.~\ref{fig:figure1} for an illustration.

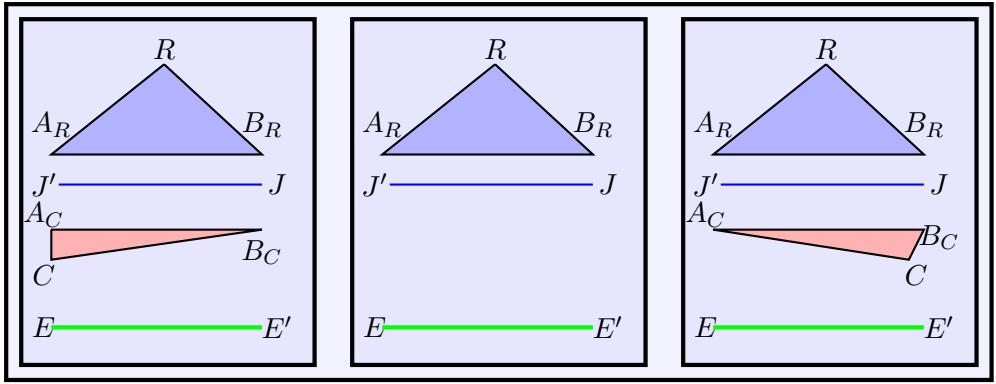
\begin{figure}
    \centering
    \begin{tikzpicture}[xscale=1,yscale=1]

\draw[ultra thick, fill=blue!5!white] (-0.1, -1) rectangle (13,4);

\draw[ultra thick, fill=blue!10!white] (0.1, -0.8) rectangle (4,3.8);

\draw [thick, fill=blue!30!white] (2,3.2) -- (0.5,2) -- (3.3,2) -- (2,3.2);
\node at (2,3.4) {$R$};
\node at (0.5, 2.4) {$A_R$};
\node at (3.3, 2.4) {$B_R$};
\draw [thick, blue] (0.6,1.6) -- (3.3,1.6);
\node at (0.4, 1.6) {$J'$};
\node at (3.5, 1.6) {$J$};
\draw [thick, fill=red!30!white] (0.5, 1) -- (3.3, 1) -- (0.5, 0.6) -- (0.5, 1);
\node at (0.4, 1.2) {$A_C$};
\node at (0.4, 0.4) {$C$};
\node at (3.3, 0.7) {$B_C$};

\draw [ultra thick, green] (0.5,-0.3) -- (3.3,-0.3);
\node at (0.4, -0.3) {$E$};
\node at (3.5, -0.3) {$E'$};

\draw[ultra thick, fill=blue!10!white] (0.1+4.4, -0.8) rectangle (4+4.4,3.8);

\draw [thick, fill=blue!30!white] (2+4.4,3.2) -- (0.5+4.4,2) -- (3.3+4.4,2) -- (2+4.4,3.2);
\node at (2+4.4,3.4) {$R$};
\node at (0.5+4.4, 2.4) {$A_R$};
\node at (3.3+4.4, 2.4) {$B_R$};
\draw [thick, blue] (0.6+4.4,1.6) -- (3.3+4.4,1.6);
\node at (0.4+4.4, 1.6) {$J'$};
\node at (3.5+4.4, 1.6) {$J$};

\draw [ultra thick, green] (0.5+4.4,-0.3) -- (3.3+4.4,-0.3);
\node at (0.4+4.4, -0.3) {$E$};
\node at (3.5+4.4, -0.3) {$E'$};

\draw[ultra thick, fill=blue!10!white] (0.1+8.8, -0.8) rectangle (4+8.8,3.8);

\draw [thick, fill=blue!30!white] (2+8.8,3.2) -- (0.5+8.8,2) -- (3.3+8.8,2) -- (2+8.8,3.2);
\node at (2+8.8,3.4) {$R$};
\node at (0.5+8.8, 2.4) {$A_R$};
\node at (3.3+8.8, 2.4) {$B_R$};
\draw [thick, blue] (0.6+8.8,1.6) -- (3.3+8.8,1.6);
\node at (0.4+8.8, 1.6) {$J'$};
\node at (3.5+8.8, 1.6) {$J$};
\draw [thick, fill=red!30!white] (0.5+8.8, 1) -- (3.3+8.8, 1) -- (3.1+8.8, 0.6) -- (0.5+8.8, 1);
\node at (0.4+8.8, 1.2) {$A_C$};
\node at (3.2+8.8, 0.4) {$C$};
\node at (3.5+8.8, 0.9) {$B_C$};

\draw [ultra thick, green] (0.5+8.8,-0.3) -- (3.3+8.8,-0.3);
\node at (0.4+8.8, -0.3) {$E$};
\node at (3.5+8.8, -0.3) {$E'$};
\end{tikzpicture}
    \caption{An illustration of the zero-cost protocol for redistributing Markov states. Left: Registers~$R A^ RB^R J J' A^C C B^C$ are in the state given in Eq.~\eqref{eq:QMC-extension} and registers~$E$ and~$E'$ contain Alice and Bob's shares of an embezzling state, respectively. Middle: Using embezzling registers, Alice and Bob have jointly ``embezzled out'' registers~$A^C C B^C$ via local unitary operations. I.e., they reverse the process of generating the state in registers~$A^C C B^C $ via embezzlement. Right: Using embezzling registers, conditioned on~$J$ and~$J'$, Alice and Bob embezzle~$\ket{\psi_j}^{A^C C B^C}$ such that registers~$C$ and~$B^C$ are with Bob and register~$A^C$ is with Alice. This step also only involves local unitary operations without any communication. }
    \label{fig:figure1}
\end{figure}

The protocol we design (for redistributing an arbitrary state) is a more sophisticated version of the above protocol. The key technique underlying this protocol is a reduction procedure using embezzling quantum states, that allows us to use a protocol due to Anshu, Jain, and Warsi~\cite{AJW18-Achievability-QSR} as a subroutine.  Let~$\sigma^{RBC}$ be a quantum Markov extension of~$\psi^{RB}$. The reduction procedure is a method which decouples~$C$ from~$RB$ when applied to~$\sigma^{RBC}$, while preserving~$\psi^{RB}$ when applied to~$\psi^{RBC}$. Preserving~$\psi^{RB}$ ensures that the reduction procedure can be implemented via local operations by Alice and Bob, without the need for any communication. Once we have a state~$\sigma^{RBC}$ such that~$\sigma^{RB} = \psi^{RB}$ and~$\sigma^{RBC} = \sigma^{RB} \tensor \sigma^C$, with the max-relative entropy and smooth hypothesis-testing relative entropy expressions as in Eq.~\eqref{eq:AJW-bound} close to those with the original states, state redistribution with the AJW protocol gives us the claimed result. Note that the reduction procedure, and in general our protocol, works for any quantum Markov extension~$\sigma^{RBC}$ of ~$\psi^{RB}$. However, in order to prove the closeness of hypothesis-testing entropy, we need to additionally assume that~$\sigma^{RBC}$ is in~$\ME^{\epsilon^2/4, \psi'}$. (See Eq.~\eqref{eq:Dh-kappa2-product} in Claim~\ref{claim-1} for a formal statement of this closeness property.) Essentially,~$\ME^{\epsilon^2/4, \psi'}$ restricts~$\sigma^{RBC}$ to quantum Markov chains for which~$\sigma_j^{B^C C}$ is close to the projection of~$\psi^{B^C C}$ on the support of~$\sigma_j^{B^C C}$ in the decomposition of~$\sigma^{RBC}$ as in Eq.~\eqref{eq:QMC-extension}.
 
To elaborate further, consider an example where~$\psi^{RBC}$ is the GHZ state~$\tfrac{1}{\sqrt{d}} \sum_{j=1}^d \ket{j}^R \ket{j}^B \ket{j}^C$. In this case, the closest Markov extension~$\sigma^{RBC}$ of~$\psi^{RB}$ is~$\tfrac{1}{d} \sum_{j=1}^{d} \density{j}^R \tensor \density{j}^B \tensor \density{j}^C$. 
A naive way to decouple register~$C$ from registers~$RB$ in~$\sigma^{RBC}$ is to coherently erase register~$C$ conditioned on register~$B$. However, the same operation applied to~$\psi^{RBC}$ changes~$\psi^{RB}$. 
To overcome this problem, first, we coherently ``measure'' register~$B$ by adding a maximally entangled state~$\ket{\Psi}^{TT'}$ and making another ``copy'' of~$\ket{j}^{B}$ in~$\Psi^{T}$. The copying is done by applying a distinct Heisenberg-Weyl operator to the state~$\Psi^T$, for each~$j \in [d]$. This operation measures register~$B$ in~$\psi^{RBC}$, keeps~$\sigma^{RBC}$ unchanged, and leaves~$\Psi^T$ in tensor product with registers~$RB$ in both~$\psi$ and~$\sigma$. Then, conditioned on register~$B$, we can coherently erase register~$C$ in~$\sigma^{RBC}$; this operation applied to~$\psi$ does not change the state~$\psi^{RB}$. Subsection~\ref{sec-GHZ} contains the complete details.
 
For a general state~$\psi^{RBC}$ with quantum Markov extension~$\sigma^{RBC}$, the isometry operator~$V_2$ can be used to transform~$\sigma^{RBC}$ to the classical-quantum state~$\sum_{j} p(j) \sigma_j^{RB^R} \tensor \density{j}^J  \tensor \sigma_j^{B^C C}$. However, we encounter an additional issue here: it may not be possible to unitarily transform all of~$\sigma_j^{B^C C}$ to a fixed state since the spectrum of~$\sigma_j^{B^C C}$ is not necessarily the same for all~$j \in [d]$. So we first ``flatten''~$\sigma_j^{B^C C}$ for each~$j$ through a unitary procedure. This task can be achieved via the technique of coherent flattening via embezzlement due to Anshu and Jain~\cite{AJ18-efficient-convex-split}. After flattening, the dimension of the support of systems~$B^C C$ no longer depends on~$j$ and so the states in registers~$B^C C$ can all be rotated to a flat state over a fixed subspace. Hence,~$B^C C$ gets decoupled from ~$RB^R J$ in the state~$\sigma$. Finally, to keep~$\psi^{RB}$ unchanged, we regenerate the system~$B^C$ via a standard embezzling technique similar to the protocol in Fig.~\ref{fig:figure1}.

\subsection{Organization of the paper}

The rest of this paper is organized as follows. In Section~\ref{sec:Preliminaries}, we present the notation and background necessary for developing the main result, namely Theorem~\ref{thm:oneshotcondmutrelent}. In section~\ref{sec-background}, we review basic concepts and results from quantum information theory. In Section~\ref{sec:QMarkov}, we define quantum Markov states and present some of their properties. We also identify a natural subset of quantum Markov states related to a given state; this subset plays a central role in the main result. 

In Section~\ref{sec-ajw}, we define the task of quantum state redistribution formally, and present two key primitives, namely Coherent Rejection Sampling (implicit in the Convex-Split Lemma) and Position-Based Decoding. We then describe how these are used by Anshu, Jain, and Warsi~\cite{AJW18-Achievability-QSR} to design a one-shot protocol for quantum state redistribution. 

Next we present some of the other components of the new protocol we develop. In Section~\ref{sec:decoupl-CQ}, we introduce a technique for decoupling classical-quantum states via embezzlement~\cite{vDH03-embezzling} and a flattening technique designed in ref.~\cite{AJ18-efficient-convex-split}.

We develop the new protocol for one-shot quantum state redistribution in Section~\ref{sec-new-protocol}. We first explain the intuition behind the protocol in detail by considering the example of the~$d$-dimensional GHZ state in Section~\ref{sec-GHZ}. We then describe the steps of the protocol for arbitrary states and analyze it in Section~\ref{sec:protocol-general}. We show how the one-shot protocol leads to the optimal communication rate for quantum state redistribution in the asymptotic i.i.d.\ case in Section~\ref{sec-aiid-qsr}. 

We conclude with a summary of the results and an outlook in Section~\ref{sec-outlook}.

Throughout Sections~\ref{sec:QMarkov}--\ref{sec:decoupl-CQ}, we provide proofs of some lemmas and theorems which are implicit in the literature. Most of these proofs are not essential for understanding the main result of this paper. The reader may safely skip the proofs if they so wish. The reader familiar with the prior work mentioned above may also start with Section~\ref{sec-new-protocol} directly, and refer to Section~\ref{sec:Preliminaries} as needed.

\section{Preliminaries}
\label{sec:Preliminaries}

\subsection{Mathematical notation and background}
\label{sec-background}

For a thorough introduction to basics of quantum information and Shannon theory, we refer the reader to the books by Watrous~\cite{W18-TQI} and Wilde~\cite{Wilde13-qit}. In this section,  we briefly review the
notation and some results that we use in this article.

For the sake of brevity, we denote the set~$\{1,2,\ldots,k\}$ by~$[k]$.
We denote physical quantum systems (``registers'') with capital letters,
like~$A$,~$B$ and~$C$. The state space corresponding to a register is
a finite-dimensional Hilbert space. We denote (finite dimensional) 
Hilbert spaces by capital script letters like~$\cH$ and $\cK$, 
and the Hilbert space corresponding to a register~$A$ by~$\cH^{A}$. We denote the dimension of the space~$\cH^{A}$ by~$\size{A}$. We sometimes refer to the space corresponding to the register~$A$ by the name of the register. 

We use the Dirac notation, i.e., ``ket'' and ``bra'', for unit vectors and
their adjoints, respectively. We denote the set of all linear operators on Hilbert 
space~$\cH$ by~$\linear(\cH)$, the set of all positive semi-definite operators
by~$\Pos(\cH)$, the set of all unitary operators by~$\unitary(\cH)$, and
the set of all quantum states (or ``density operators'') over~$\cH$
by~$\qstate(\cH)$. The identity operator on space~$\cH$ or register~$A$, is denoted
by~$\id^\cH$ or~$\id^A$, respectively. Similarly, we use superscripts to indicate the registers on which an operator acts. We say a positive semi-definite operator~$M \in \Pos(\cH)$ is a \emph{measurement operator\/} if~$M \preceq\id^{\cH}$, where~$\preceq$ denotes \emph{L{\"o}wner order} for Hermitian operators. 

\suppress{
Let~$T$ be a register with~$\size{T} = d \ge 1$. For~$a,b \in [d]$, the \emph{Heisenberg-Weyl\/} operator~$P_{a,b} \in \unitary \br{\cH^{T}}$ is defined as
\[
P_{a,b} \quad \coloneqq \quad \sum_{t = 1}^{d} \exp\br{ \frac{2\pi \complexi tb}{d} } \; \ketbra{t \oplus a}{t} \enspace,
\]
where the addition is cyclic, i.e., $t \oplus a = t+a - \floor{(t+a-1)/d} d$. We write~$P_a$ to denote~$P_{a,d}$. 
}

Let~$T$ be a register with~$\size{T} = d \ge 1$. For~$a \in [d]$, we define the operator~$P_{a} \in \unitary \br{\cH^{T}}$ as
\[
P_{a} \quad \coloneqq \quad \sum_{t = 1}^{d} \ketbra{t \oplus a}{t} \enspace,
\]
where the addition~`$\oplus$' is cyclic, i.e., $t \oplus a = t+a - d \floor{(t+a-1)/d}$. This is the~$a$-th power of the generalized Pauli operator (also called a \emph{Heisenberg-Weyl\/} operator).

We denote quantum states by lowercase Greek letters like~$\rho, \sigma$. 
We use the notation~$\rho^{A}$ to indicate that register~$A$ is in quantum state~$\rho$.
We denote the \emph{partial trace\/} operation over register~$A$ by~$\trace_{A}$. When it is clear from the context, we also use~$\rho^{B}$ to denote the partial trace of a state~$\rho^{AB}$ over~$B$.  We 
say~$\rho^{AB}$ is an \emph{extension\/} of~$\sigma^{A}$ 
if~$\trace_{B} (\rho^{AB} ) = \sigma^{A}$. A \emph{purification} of a quantum 
state~$\rho$ is an extension of~$\rho$ with rank one. For the Hilbert space~$\complex^S$ for some set~$S$, we refer to the basis~$\set{ \ket{x} : x \in S}$ as the canonical basis for the space. We say the register~$X$ 
is \emph{classical\/} in a quantum state~$\rho^{XB}$ if~$\rho^{XB}$ is block-diagonal 
in the canonical basis of~$X$, i.e.,~$\rho^{XB} = \sum_{x} p(x) \density{x}^{X} \tensor \rho^{B}_{x}$ for some probability distribution~$p$ on~$X$.
For a non-trivial register~$B$, we say~$\rho^{XB}$ is a \emph{classical-quantum} state if~$X$ is classical in~$\rho^{XB}$.
We say a unitary operator~$U^{AB} \in \unitary(\cH^{A} \tensor \cH^{B})$ 
is \emph{read-only} on register~$A$ if it is block-diagonal in the canonical basis of~$A$, i.e.,~$U^{AB} = \sum_{a} \density{a}^A \tensor U_{a}^{B}$ where each~$U_{a}^{B}$ is a unitary operator.

The \emph{trace norm\/} (Schatten 1 norm) of an operator~$M \in \linear(\cH)$ is the sum of its 
singular values and we denote it by~$\trnorm{M}$. The \emph{trace distance\/} between~$\rho$ and~$\sigma$
is induced by trace norm. The following theorem is a well-known property of 
trace norm (see, e.g.,~\cite[Theorem~3.4, page~128]{W18-TQI}). 
\begin{theorem} [Holevo-Helstrom \cite{Hel67-detection,Hol72-decision}]
\label{prop-trnorm-povm}
For any pair of quantum states~$\rho, \sigma \in \qstate(\cH)$,
\[
\trnorm{\rho - \sigma} \quad = \quad 2 \; \max \;
    \set{ \; \abs{ \trace( \Pi \rho) - \trace( \Pi \sigma) } : \Pi\preceq \id, \Pi \in \Pos(\cH) } \enspace.
\]
\end{theorem}

\begin{lemma} [Gentle Measurement~\cite{Win99-GML,ON07-GML}]
\label{prop-gentl-meas}
Let~$\epsilon \in [0,1]$,~$\rho \in \qstate(\cH)$ and~$\Pi \in \Pos(\cH)$ be a measurement operator such that~$\trace (\Pi \rho) \ge 1-\epsilon$. Then,
\begin{equation*}
    \trnorm{ \frac{\Pi \rho \Pi}{\trace (\Pi \rho)}- \rho } \quad \le \quad 2 \sqrt{\epsilon} \enspace.
\end{equation*}
\end{lemma}

The \emph{fidelity} between two sub-normalized states~$\rho$ and~$\sigma$
is defined as
\[
\rF(\rho,\sigma) \quad \coloneqq \quad \trace \sqrt{\sqrt{\rho}\ \sigma \sqrt{\rho}} + 
\sqrt{ \left( 1- \trace(\rho)\right)\left( 1- \trace(\sigma)\right)}\enspace.
\]
Fidelity can be used to define a useful metric called the 
\emph{purified distance\/}~\cite{Rast02-cloning,Rast03-cloning,Rast06-Sine,GLN05-dist-meas,TCR10-duality} between
quantum states:
\[
\rP(\rho,\sigma)\quad\coloneqq\quad\sqrt{1-\rF(\rho,\sigma)^{2}}\enspace.
\]

Purified distance and trace distance are related to each other as follows (see, e.g.,~\cite[Theorem~3.33, page 161]{W18-TQI}):
\begin{theorem}[Fuchs and van de Graaf inequality \cite{FG99-dist-meas}]
\label{prop-FvdG}
For any pair of quantum states~$\rho, \sigma \in \qstate(\cH)$,
\[
1 - \sqrt{1- \rP(\rho,\sigma)^2}
    \quad \leq \quad \frac{1}{2} \trnorm{\rho - \sigma}
    \quad \leq \quad \rP(\rho, \sigma) \enspace.
\]
\end{theorem}

For a quantum state~$\rho \in \qstate(\cH)$ and~$\epsilon\in[0,1]$, we
define
\[
\sB^{\epsilon}(\rho) \quad \coloneqq \quad
    \set{ \widetilde{\rho} \in \qstate(\cH) : ~ \rP(\rho, \widetilde{\rho}) 
    \leq \epsilon }
\]
as the ball of quantum states that are within purified 
distance~$\epsilon$ of~$\rho$. Note that in some works, the states in the set~$\sB^{\epsilon}(\rho)$ are allowed to be sub-normalized. Here, we require the states in the ball to have trace equal to one. 

\begin{theorem}[Uhlmann \cite{Uhl76-Uhlmann-thm}]
Consider quantum states~$\rho^{A},\sigma^{A} \in \qstate (\cH^{A})$. Suppose~$\ket{\xi}^{AB},\ket{\theta}^{AB} \in \qstate(\cH^{A} \tensor \cH^B)$ are arbitrary purifications of~$\rho^{A}$ and~$\sigma^{A}$, respectively. Then, there exists some unitary operator~$V^{B} \in \unitary(\cH^{B})$ such that
\begin{equation*}
    \rP \mleft( \ket{\xi}^{AB},\left(\id \tensor V^{B} \right) \ket{\theta}^{AB} \mright) \quad = \quad \rP(\rho^{A}, \sigma^{A}) \enspace.
\end{equation*}
\end{theorem}

Let~$\rho \in \qstate(\cH)$ be a quantum state over the Hilbert space~$\cH$. The
\emph{von Neumann entropy\/} of~$\rho$ is defined as
\[
\rS(\rho) \quad \coloneqq \quad - \trace \left( \rho \log \rho \right) \enspace.
\]
This coincides with Shannon entropy for a classical state.
The \emph{relative entropy\/} of two quantum states $\rho, \sigma
\in \qstate(\cH)$ is defined as
\[
\rD(\rho \| \sigma) \quad \coloneqq \quad \trace 
\left( \rho \left( \log \rho - \log \sigma \right) \right) \enspace,
\]
when~$\support(\rho) \subseteq \support(\sigma)$, and is~$\infty$
otherwise. The \emph{max-relative entropy\/}~\cite{Dat09-min-max-ent} of~$\rho$ with respect 
to~$\sigma$ is defined as
\[
\Dmax( \rho \| \sigma) \quad \coloneqq \quad
    \min \{ \lambda : \rho \leq 2^{\lambda} \sigma \} \enspace,
\]
when~$\support(\rho) \subseteq \support(\sigma)$, and is~$\infty$
otherwise.
The following proposition bounds purified distance in terms of max-relative entropy. It is
a special case of the monotonicity of \emph{minimal quantum~$\alpha$-R\'enyi divergence} in~$\alpha$ (see, e.g.,~\cite[Corollary 4.2, page 56]{Tom15-QIP-finite-resource}) obtained by considering~$\alpha = 1/2$ and~$\alpha \rightarrow \infty$.
\begin{proposition}[\cite{MDSFT13-quantum-renyi}]
\label{prop:Dmax-Pdist}
	Let~$\cH$ be a Hilbert space, and let~$\rho,\sigma\in\qstate(\cH)$ be 
	quantum states over~$\cH$. It holds that
	\[
	\rP( \rho, \sigma) \quad \leq \quad \sqrt{1-2^{-\Dmax (\rho \| \sigma)}} \enspace.
	\]
\end{proposition}
The above property also implies the \emph{Pinsker inequality\/}.
For~$\epsilon \in [0,1]$, the \emph{$\epsilon$-smooth max-relative entropy\/}~\cite{Dat09-min-max-ent}
of~$\rho$ with respect to~$\sigma$ is defined as
\[
\Dmax^{\epsilon} (\rho \| \sigma) 
\quad \coloneqq \quad \min_{\rho' \in \sB^{\epsilon}(\rho)} \Dmax (\rho' \| \sigma) \enspace. 
\]
For~$\epsilon \in [0,1]$, the~$\epsilon$-\emph{hypothesis testing relative entropy}~\cite{BD10-QCapacity-correlated, BD11-Entanglement-manipulation, WR12-CQ-capacity-Hypothesis-testing} of~$\rho$ with 
respect to~$\sigma$ is defined as
\[
\Dh^{\epsilon} \left(\rho \| \sigma \right) \quad \coloneqq \quad 
\sup_{0 \preceq \Pi \preceq \id , \, \trace(\Pi \rho) \geq 1-\epsilon}
\log \left( \frac{1}{\trace(\Pi \sigma)} \right) \enspace.
\]
Smooth max-relative entropy and hypothesis testing relative entropy both converge to relative entropy in the asymptotic and i.i.d.\ setting~\cite{TCR09-QAEP, Tom12-thesis, KMF12-state-discr}. The following proposition gives upper and lower bounds for the convergence of these quantities for finite~$n$; these bounds are tight up to the second order additive term.
\begin{theorem} [\cite{TH13-Hierarchy},\cite{Li14-asym-Dh}]
\label{prop-QAEP}
Let~$\epsilon \in (0,1)$ and~$n$ be an integer. Consider quantum states~$\rho, \sigma \in \qstate (\cH)$. Define~$\rV(\rho \| \sigma) \eqdef \trace (\rho (\log \rho - \log \sigma)^2) -(\rD(\rho \| \sigma))^2$ and~$\Normal(x) \eqdef \int_{-\infty}^{x} \tfrac{\exp(-x^2/2)}{\sqrt{2\pi}} \diff x$. It holds that
\begin{equation}
\label{eq:QAEP-Dmax}
    \Dmax^{\epsilon}\mleft(\rho^{\tensor n} \| \sigma^{\tensor n} \mright) \quad = \quad n\rD(\rho \| \sigma) - \sqrt{n\rV(\rho \| \sigma)} \Normal^{-1}(\epsilon^2) + \Order(\log n) - \Order(\log(1-\epsilon) )\enspace,
\end{equation}
and
\begin{equation}
    \Dh^{\epsilon} \mleft( \rho^{\tensor n} \| \sigma^{\tensor n} \mright) \quad = \quad n\rD(\rho \| \sigma) + \sqrt{n\rV(\rho \| \sigma)} \Normal^{-1}(\epsilon) + \Order(\log n)\enspace.
\end{equation}
\end{theorem}
Note that Eq.~\eqref{eq:QAEP-Dmax} has an additional~$\Order(\log (1-\epsilon))$ term as compared to the original statement in ref.~\cite{TH13-Hierarchy} because we only allow the normalized states in~$\sB^{\epsilon}(\rho)$.
We also need the following property due to Anshu, Berta, Jain, and Tomamichel~\cite[Theorem 2]{ABJM20-Partially-smooth}. The original statement involves a  minimization over all $\sigma_B$ on both sides of the inequality, but the proof works for any fixed~$\sigma_B$.
\begin{theorem}[\cite{ABJM20-Partially-smooth}, Theorem~2]
\label{prop-Dmax-psmooth}
Let~$\epsilon, \delta \in (0,1)$ such that~$0 \le 2\epsilon + \delta \le 1$. Consider quantum states~$\sigma^{B} \in \qstate(\cH^{B})$ and~$\rho^{AB} \in \qstate(\cH^{AB})$. We have
\begin{equation}
    \inf_{\substack{\barrho \in \sB^{2\epsilon + \delta}(\rho^{AB}) \\ \barrho^{A} = \rho^{A}}} \Dmax \mleft( \barrho^{AB} \| \rho^{A} \tensor \sigma^{B} \mright) \quad \leq \quad  \Dmax^{\epsilon}\mleft(\rho^{AB} \| \rho^{A} \tensor \sigma^{B} \mright) + \log \frac{8+ \delta^2}{\delta^2} \enspace.
\end{equation}
\end{theorem}

Suppose that $\rho^{AB} \in \qstate(\cH^{A} \otimes \cH^{B})$ is the joint 
state of registers $A$ and $B$, then the \emph{mutual information\/} of
$A$ and $B$ is denoted by
\[
\rI(A:B)_{\rho} \quad \coloneqq \quad \rD \!\left( \rho^{AB} \|~\rho^{A} \tensor \rho^{B} \right) \enspace.
\]
When the state is clear from the context, the subscript~$\rho$ may be omitted.
Let~$\rho^{RBC}\in \qstate(\cH^{RBC})$ be a 
tripartite quantum state. The \emph{conditional mutual information\/} 
of~$R$ and~$C$ given~$B$ is defined as
\[
\rI(R:C \,|\, B) \quad \coloneqq \quad \rI(RB:C) - \rI(B:C) \enspace.
\]
For the state~$\rho^{AB} \in \qstate(\cH^{A} \otimes \cH^{B})$, the \emph{max-information\/} register~$B$ has about register~$A$ is defined as
\[
\Imax(A:B)_{\rho} \quad \coloneqq \quad 
\min_{\sigma^{B} \in \qstate(\cH^{B})} \Dmax \left(\rho^{AB} \|~\rho^{A} \tensor \sigma^{B} \right) \enspace.
\]
For~$\epsilon \in [0,1]$, the \emph{$\epsilon$-smooth max-information\/} register~$B$ has about register~$A$ 
in the state~$\rho^{AB} \in \qstate(\cH^{A} \otimes \cH^{B})$ is defined as
\[
\Imax^{\epsilon}(A:B)_{\rho} \quad \coloneqq \quad \min_{\rho' \in \sB^{\epsilon}(\rho^{AB})}
\Imax (A:B)_{\rho'} \enspace.
\]

\subsection{Quantum Markov states}
\label{sec:QMarkov}

A tripartite quantum state~$\sigma^{RBC} \in \qstate(\cH^{RBC})$ is 
called a \emph{quantum Markov state\/} of the form~$R \!-\! B \!-\! C$ if there exists a quantum 
operation~$\Lambda : \linear \! \left(\cH^{B}\right) \rightarrow \linear \!\left(\cH^{BC}\right)$ 
such that~$(\id \tensor \Lambda)(\sigma^{RB}) = \sigma^{RBC}$. This is equivalent to the condition that~$\rI(R:C \, | \,B)_\sigma = ~ 0$, and is the quantum analogue of the notion of \emph{Markov chains\/} for classical registers. 
Classical registers~$YXM$ form a \emph{Markov chain\/} in this order (denoted as~$Y\!-\!X\!-\!M$) if registers~$Y$ and~$M$ are independent  given~$X$. 
Hayden, Josza, Petz, and Winter~\cite{HJPW04-Qntm-Markov-state} showed that an analogous property holds for quantum Markov states. 
\begin{theorem}[\cite{HJPW04-Qntm-Markov-state}]
A state~$\sigma^{RBC} \in \qstate(\cH^{R} \tensor \cH^{B} \tensor \cH^{C})$ is a quantum Markov state of the form~$R \!-\! B \!-\! C$ if and only if there is a decomposition of the space~$ \cH^{B}$ into a direct sum of tensor products as
\begin{equation}
\label{eq:markov-reg-decomp}
	\cH^{B} \quad = \quad \bigoplus_{j} \cH^{B_{j}^{R}} \tensor \cH^{B_{j}^{C}} \enspace,
\end{equation}
such that
\begin{equation}
\label{eq:markov-state-decomp}
   \sigma^{RBC} \quad = \quad \bigoplus_{j} p(j) \, \sigma_j^{RB_{j}^{R}} \tensor
   \sigma_j^{B_{j}^{C}C} \enspace, 
\end{equation}
where~$\sigma_j^{RB_{j}^{R}} \in \qstate \left(\cH^{R} \tensor \cH^{B_{j}^{R}} \right)$,~$\sigma_j^{B_{j}^{C}C} \in \qstate \left(\cH^{B_{j}^{C}} \tensor \cH^{C} \right)$ and~$p$ is a probability distribution over the direct summands.
\end{theorem}

For a state~$\psi^{RBC}$, we say that~$\sigma^{RBC}$ is a \emph{Markov extension\/} of~$\psi^{RB}$ if~$\sigma^{RB} = \psi^{RB}$ and~$\sigma^{RBC}$ is a Markov state. We denote the set of all Markov extensions of~$\psi^{RB}$ by~$\QMC^{\psi}$. Note that~$\QMC^{\psi}$ is non-empty, as it contains the state~$\sigma^{RBC} \eqdef \psi^{RB} \tensor \psi^C $. The following lemma relates the quantum conditional mutual information to quantum Markov extensions. The proof of this lemma is implicit in ref.~\cite[Lemma 1]{BertaSW15}, but we provide a proof here for completeness.

\begin{lemma}[Implicit in~\cite{BertaSW15}, Lemma 1]
\label{lem:condmutrelent}
For any tripartite quantum state $\psi^{RBC}$, and any quantum Markov extension~$\sigma^{RBC} \in \QMC^{\psi}$, it holds that
\begin{equation*}
\rI(R:C \,|\, B)_{\psi} \quad = \quad \rD\br{\psi^{RBC}\|\sigma^{RBC}}-\rD\br{\psi^{BC}\|\sigma^{BC}} \enspace.
\end{equation*}
\end{lemma}

\begin{proof}
    For sake of clarity, in this proof, we suppress tensor products with the identity in expressions involving sums or products of quantum states over different sequences of registers. For example, we write~$\omega^{XY} + \tau^{YZ}$ to represent the sum~$\omega^{XY} \tensor \id^Z + \id^X \tensor \tau^{YZ}$, and~$\omega^{XY} \tau^{YZ}$ to represent the product~$\big( \omega^{XY} \tensor \id^Z \big) \big( \id^X \tensor \tau^{YZ} \big)$. All the expressions involving entropy and mutual information are with respect to the state~$\psi$.

Consider any quantum Markov chain $\sigma^{RBC}$ satisfying $\sigma^{RB}=\psi^{RB}$. From Eq.~\eqref{eq:markov-state-decomp}, we have
$$\log\sigma^{RBC} \quad = \quad  \bigoplus_{j} \br{\log \mleft( p(j)\sigma_j^{RB_{j}^{R}}\mright) + \log\sigma_j^{B_{j}^{C}C}} \enspace,$$
and similarly, 
$$\log\sigma^{BC} \quad = \quad  \bigoplus_{j} \br{\log \mleft( p(j)\sigma_j^{B_{j}^{R}} \mright) + \log\sigma_j^{B_{j}^{C}C}} \enspace.$$
Thus, we can evaluate
\begin{IEEEeqnarray*}{rCl}
\IEEEeqnarraymulticol{3}{l}{
\rD\br{\psi^{RBC}\|\sigma^{RBC}}-\rD\br{\psi^{BC}\|\sigma^{BC}} }\\* \qquad  
& = & \quad \trace \br{\psi^{RBC}\log\psi^{RBC}} - \trace \br{\psi^{RBC}\log\sigma^{RBC}} - \trace \br{\psi^{BC}\log\psi^{BC}} + \trace \br{\psi^{BC}\log\sigma^{BC}}\\
& = & \quad \rS(BC) - \rS(RBC) - \sum_j \trace \br{\psi^{RBC}\log \br{ p(j)\sigma_j^{RB_{j}^{R}}} } - \sum_{j}\trace \br{\psi^{RBC}\log\sigma_j^{B_{j}^{C}C}} \\*
 && \quad + \> \sum_{j} \trace \br{\psi^{BC}\log \br{p(j) \sigma_j^{B_{j}^{R}}} } + \sum_{j} \trace \br{\psi^{BC}\log\sigma_j^{B_{j}^{C}C}} \enspace.
\end{IEEEeqnarray*}
Since $\trace \br{\psi^{RBC}\log\sigma_j^{B_{j}^{C}C}}= \trace \br{\psi^{BC}\log\sigma_j^{B_{j}^{C}C}}$, the above equation can be simplified to obtain
\begin{IEEEeqnarray*}{rCl}
\IEEEeqnarraymulticol{3}{l}{
\rD\br{\psi^{RBC}\|\sigma^{RBC}}-\rD\br{\psi^{BC}\|\sigma^{BC}}}\\* \qquad
& = & \quad \rS(BC) - \rS(RBC) -\sum_{j} \trace \br{ \psi^{RBC} \log \br{p(j) \sigma_j^{RB_{j}^{R}}} } + \sum_{j} \trace \br{\psi^{BC}\log \br{ p(j) \sigma_j^{B_{j}^{R}}} } \\
& = & \quad \rS(BC) - \rS(RBC) - \trace \br{\psi^{RBC}\log \br{ \bigoplus_{j} p(j) \sigma_j^{RB_{j}^{R}}} } + \trace \br{ \psi^{BC} \log \br{ \bigoplus_{j} p(j) \sigma_j^{B_{j}^{R}}} }\\
& = & \quad \rS(BC) - \rS(RBC) - \trace \br{\psi^{RBC} \log \bigoplus_{j} \br{ p(j) \sigma_j^{RB_{j}^{R}} \otimes \sigma_j^{B_{j}^{C}}} } \\
&& \quad + \> \trace \br{\psi^{BC} \log \bigoplus_{j} \br{ p(j) \sigma_j^{B_{j}^{R}} \otimes \sigma_j^{B_{j}^{C}}} } \enspace,
\end{IEEEeqnarray*}
where the last equality above follows by noting that
\[
\trace \br{\psi^{RBC}\log\sigma_j^{B_{j}^{C}}} \quad = \quad \trace \br{ \psi^{BC} \log \sigma_j^{B_{j}^{C}} } \enspace.
\]
Since $\psi^{RB}=\sigma^{RB}$, we get that
\begin{IEEEeqnarray*}{rCl}
\rD\br{\psi^{RBC}\|\sigma^{RBC}}-\rD\br{\psi^{BC}\|\sigma^{BC}} \quad
& = & \quad \rS(BC) - \rS(RBC) - \trace \br{\psi^{RBC}\log \sigma^{RB}} + \trace \br{\psi^{BC}\log \sigma^{B}}\\
& = & \quad \rS(BC) - \rS(RBC) - \trace \br{\psi^{RB} \log \psi^{RB}} + \trace \br{\psi^{B}\log \psi^{B}}\\
& = & \quad \rS(BC) - \rS(RBC) + \rS(RB) - \rS(B) \\
& = & \quad \rI\br{R:C \,| \,B} \enspace.
\end{IEEEeqnarray*}
This completes the proof.
\end{proof}

For a Markov extension~$\sigma \in \QMC^{\psi}$, let~$\Pi_j^\sigma$ be the orthogonal projection operator onto the $j$-th subspace of the register~$B$ given by the decomposition corresponding to the Markov state~$\sigma$ as described above. In other words, $\Pi_j^\sigma$ is the projection onto the Hilbert space~$\cH^{B_j^{R}} \tensor \cH^{B_j^{C}}$ in Eq.~\eqref{eq:markov-reg-decomp}. For a quantum
state~$\psi^{RBC}$, we define
\begin{equation}
    \ME^{\epsilon,\psi} \quad \coloneqq \quad \set{ \left. \sigma \in \QMC^{\psi}  ~ \right| ~ \text{ for all } j, ~ \sigma_j^{B_{j}^{C} C} \in \sB^{\epsilon}\mleft(\trace_{B_j^{R}}\left[( \Pi_j^\sigma \tensor \id) \psi^{BC} (\Pi_j^\sigma \tensor \id)\right]\mright) } \enspace. \label{eq:ME-def}
\end{equation}
Informally, this is the subset of Markov extensions~$\sigma$ of~$\psi$ such that the restrictions of~$\sigma$ and~$\psi$ to the~$j$-th subspace in the decomposition of~$\sigma$ agree well on the registers~$B_{j}^{C} C$. Again, the state~$\sigma^{RBC} \eqdef \psi^{RB} \tensor \psi^C $ belongs to~$\ME^{\epsilon,\psi}$ for every~$\epsilon \ge 0$, so the set is non-empty.

\suppress{
If~$\rho^{RBC}$ is a classical-quantum state, classical in~$B$, then
\[
\rI(R:C \, | \, B) \quad = \quad \min_{\substack{\sigma^{RBC} : \sigma^{RB} = \rho^{RB} \\ \exists \Lambda \, : \, (\id \tensor \Lambda)(\rho^{RB}) = \sigma^{RBC}}} \rD(\rho^{RBC} \| \sigma^{RBC}) \enspace.
\]
However, this equality breaks down if register~$B$ is no more classical. For a tripartite quantum state~$\rho^{RBC}$, we define
\[
 \markdist(R:C|B)_{\rho} 
 \quad \coloneqq \quad \min_{\substack{\sigma^{RBC} : \sigma^{RB} = \rho^{RB} \\ \exists \Lambda \, : \, (\id \tensor \Lambda)(\rho^{RB}) = \sigma^{RBC}}} \rD(\rho^{RBC} \| \sigma^{RBC}) \enspace,
 \]
and
\[
\maxmarkdist(R:C|B)_{\rho} 
\quad \coloneqq \quad \min_{\substack{\sigma^{RBC} : \sigma^{RB} = \rho^{RB} \\ \exists \Lambda \, : \, (\id \tensor \Lambda)(\rho^{RB}) = \sigma^{RBC}}} \Dmax(\rho^{RBC} \| \sigma^{RBC}) \enspace.
\]
}

\subsection{Quantum state redistribution}
\label{sec-ajw}

Consider a pure state~$\ket{\psi}^{RABC}$ shared between 
Ref ($R$), Alice ($AC$) and Bob ($B$). In an~$\epsilon$-error 
\emph{quantum state redistribution} protocol, Alice and Bob share an 
entangled state~$\ket{\theta}^{E_A E_B}$, where register~$E_A$ is with Alice and register~$E_B$ with Bob. Alice applies an encoding 
operation~$\cE : \linear(\cH^{ACE_A}) \rightarrow \linear(\cH^{AQ})$, and sends the 
register~$Q$ to Bob.
Then, Bob applies a decoding operation~$\cD : \linear(\cH^{QBE_B}) \rightarrow 
\linear(\cH^{BC})$. The output of the protocol is the state~$\phi^{RABC}$ with the property that~$\rP(\psi^{RABC}, \phi^{RABC}) \leq \epsilon$. The communication cost of the protocol is~$\log \size{Q}$. 

To derive the bound in Theorem~\ref{thm:oneshotcondmutrelent}, we use a protocol due to Anshu, Jain, and Warsi~\cite{AJW18-Achievability-QSR}, which we call the AJW protocol in the sequel. 
%
%
The AJW protocol is based on the \emph{Convex-Split Lemma\/} introduced by Anshu, Devabathini, and Jain~\cite{ADJ17-coherent-rejection},
and the technique of \emph{Position-Based Decoding} introduced by Anshu, Jain, and Warsi~\cite{AJW19-building-blck-noisy-Qtm}. 

Let~$n$ be an integer, ~$\rho^{AB} \in \qstate(\cH^{AB})$ and~$\sigma^{B} \in \qstate(\cH^{B})$. Consider the quantum state~$\tau^{AB_1 \ldots B_n}$ derived by adding~$n-1$ independent copies of~$\sigma^{B}$ in tensor product with~$\rho^{AB}$ and swapping the~$(i-1)$-th copy of~$\sigma^{B}$ with~$\rho^{B}$ for uniformly random~$i \in [n-1]$. The convex-split lemma states that the state~$\tau^{AB_1\ldots B_n}$ is almost indistinguishable from the product state~$\rho^{A} \tensor (\sigma^{B})^{\tensor n}$, provided that~$n$ is large enough.

\begin{lemma} [Convex-Split Lemma~\cite{ADJ17-coherent-rejection}]
 	\label{lemma-convex-split}
 	Let~$\rho^{AB} \in \qstate(\cH^{AB})$ and~$\sigma^{B} \in \qstate(\cH^{B})$ be 
 	quantum states	with~$\Dmax(\rho^{AB} \| \rho^{A} \tensor \sigma^{B})=k$ for 
 	some finite number~$k$. Let~$\delta > 0$ and~$n \eqdef \ceil{\tfrac{2^k}{\delta}}$. 
 	Define the following states on~$n+1$ registers~$A,B_1,B_2, \dotsc, B_n$~:
 	\begin{align*}
 	\tau^{AB_1B_2 \dotsb B_n} \quad & \coloneqq \quad \frac{1}{n} \sum_{j=1}^{n} 
 	\rho^{AB_j} \tensor \sigma^{B_1} \tensor \cdots \tensor 
 	\sigma^{B_{j-1}} \tensor \sigma^{B_{j+1}} \tensor \cdots \tensor 
 	\sigma^{B_{n}} \enspace, \qquad \text{and} \\
 	\ttau^{A B_1 B_2 \dotsb B_n} \quad & \coloneqq \quad \rho^{A} \tensor \sigma^{B_1} \tensor \dotsb \tensor \sigma^{B_n} \enspace,
 	\end{align*}
where for all~$i \in [n]$, we have~$\size{B_i} = \size{B}$, $\rho^{AB_i}=\rho^{AB}$, and~$\sigma^{B_i}=\sigma^B$. 
Then, we have
 	\begin{align*}
 	\rP\left(\tau^{A B_1 \dotsb B_n}, \; \ttau^{A B_1 \dotsb B_n} \right) 
 	\quad \le \quad \sqrt{\delta} \enspace.
 	\end{align*}
 \end{lemma}
We may think of the Convex-Split Lemma as providing a sufficient condition under which the correlations between registers~$A$ and~$B$ in~$\rho$ can be ``hidden'' by taking a certain convex combination of quantum states. A dual problem is to find conditions sufficient for \emph{identifying\/} the location of desired correlations in a convex combination. This task is achievable via the position-based decoding technique, which in turn uses quantum hypothesis testing.  

\begin{lemma}[Position-Based Decoding~\cite{AJW19-building-blck-noisy-Qtm}]
Let~$\epsilon >0$, and~$\rho^{AB} \in \qstate(\cH^{AB})$ and~$\sigma^{B} \in \qstate(\cH^{B})$ be quantum states such that~$\support(\rho^{B}) \subseteq \support(\sigma^{B})$. Let~$n \coloneqq \ceil{\epsilon \, 2^{\Dh^{\epsilon}(\rho^{AB} \| \rho^{A} \tensor \sigma^{B})}}$, and for every~$j \in [n]$,
\begin{equation*}
    \tau_j^{AB_1\ldots B_n} \quad \coloneqq \quad \rho^{AB_j} \tensor  \sigma^{B_1} \tensor \cdots \tensor 
 	\sigma^{B_{j-1}} \tensor \sigma^{B_{j+1}} \tensor \cdots \tensor 
 	\sigma^{B_{n}} \enspace.
\end{equation*}
There exists a measurement~$( \Lambda_j : j \in [n+1])$ on registers~$A B_1 B_2 \dotsb B_{n}$, i.e., operators~$\Lambda_i \succeq 0$ with
$$\sum_{j=1}^{n+1} \Lambda_j \quad = \quad \id \enspace, $$ 
such that for all~$j \in [n]$,
\begin{equation*}
    \trace \mleft[ \Lambda_j \tau_j^{AB_1\ldots B_n} \mright] \quad \geq \quad 1-6\epsilon \enspace.
\end{equation*}
\end{lemma}
The above statement is slightly different from the one in ref.~\cite{AJW19-building-blck-noisy-Qtm} because of a minor difference in defining quantum hypothesis testing relative entropy.

 Let~$\ket{\psi}^{RABC}$ be a quantum state shared between Alice, Bob, and Ref where registers~$AC$ are with Alice, register~$B$ is with Bob and register~$R$ is with Ref, and~$\psi'^{RBC} \in \sB^{\epsilon}(\psi^{RBC})$. The AJW protocol works as follows. 
 
\paragraph{The AJW protocol:} 
\begin{enumerate}
\item 
Alice and Bob initially share~$m \eqdef \ceil{2^{\beta} / \epsilon^2 }$ copies of a purification~$\ket{\sigma}^{LC}$ of~$\sigma^{C}$ where~$\beta \eqdef \Dmax \mleft( \psi'^{RBC} \| \psi'^{RB} \tensor \sigma^{C} \mright)$. Their global state is~$\ket{\psi}^{RABC} \tensor \ket{\sigma}^{L_1C_1} \tensor \ldots \tensor \ket{\sigma}^{L_m C_m}$, where~$\size{L_i} = \size{L}$ and~$\size{C_i} = \size{C}$ for all~$i \in [m]$. The registers~$AC L_1 L_2 \dotsb L_m$ are with Alice and the registers~$B C_1 C_2 \dotsb C_m$ are with Bob.

\item
Let~$b$ be the smallest integer such that~$\log b \ge \Dh^{\epsilon^2}(\psi'^{BC} \| \psi'^{B} \tensor \sigma^{C}) - \log \frac{1}{\epsilon^2} \;$. By performing a suitable isometry on her registers, Alice transforms the global state into a state close to the state
\begin{align*}
     \frac{1}{m} \sum_{j=1}^{m} \ket{\floor{(j-1)/b}}^{J_1} & \ket{j-1 \tpmod{b}}^{J_2} \ket{0}^{L_j} \ket{\psi}^{RABC_j} \\
     & \mbox{} \tensor \ket{\sigma}^{L_1C_1} \tensor \dotsb \tensor \ket{\sigma}^{L_{j-1}C_{j-1}} \tensor \ket{\sigma}^{L_{j+1}C_{j+1}} \tensor \dotsb \tensor  \ket{\sigma}^{L_{m} C_{m}} \enspace.
\end{align*}
This is possible due to the Uhlmann theorem, the Convex-Split Lemma, and the choice of~$m$.

\item
Alice sends register~$J_1$ to Bob with communication cost at most~$(\log m - \log b)/2$ using superdense coding. 

\item 
Then, for each~$j_2 \in [b]$, Bob swaps registers~$C_{j_2}$ and~$C_{j_2+bj_1}$, conditioned on register~$J_1$ being in state~$\ket{j_1}$. At this point, registers~$RBC_1 \ldots C_b$ are in a state close to
\begin{align*}
    \frac{1}{b} \sum_{j_2=1}^{b} \psi^{RBC_{j_2}} \tensor \sigma^{C_1} \tensor \ldots \tensor \sigma^{C_{j_2-1}} \tensor \sigma^{C_{j_2+1}} \tensor \ldots \tensor  \sigma^{C_{b}} \enspace.
\end{align*}

\item
Then, Bob uses position-based decoding to determine the index~$j_2$ 
for which register~$C_{j_2}$ is correlated with registers~$RB$. This 
is possible by the choice of~$b$.

\item
Since the state over registers~$RBC_{j_2}$ is close to~$\psi^{RBC}$, and it is in tensor product with the state over registers~$C_1 \dotsb C_{j_{2}-1} C_{j_{2}+1} \dotsb C_{b}$, the register purifying registers~$RBC_{j_2}$ is with Alice. She transforms the purifying registers to the register~$A$ such that the final state over registers~$RABC_{j_2}$ is close to~$\psi^{RABC}$.

\end{enumerate}

The following theorem states the communication cost and  the error in the final state of the above protocol. 

\begin{theorem} [\cite{AJW18-Achievability-QSR}]
\label{thm-QSR-AJW}
	 Let~$\epsilon \in (0,1)$, and~$\ket{\psi}^{RABC}$ be a pure quantum state shared by Ref~$(R)$, Alice~$(AC)$ and Bob~$(B)$. There is a quantum state redistribution protocol for~$\ket{\psi}^{RABC}$ which outputs a state~$\phi^{RABC} \in \sB^{9\epsilon}(\psi^{RABC})$. Moreover, the number of qubits sent by Alice to Bob in the protocol is bounded from above by
	\[
	\frac{1}{2} \inf_{\sigma^{C}} \inf_{\psi' \in \sB^{\epsilon}(\psi^{RBC})} \left(  \Dmax \mleft( \psi'^{RBC} \mright\| \mleft. \psi'^{RB} \tensor \sigma^{C} \mright) - \Dh^{\epsilon^{2}} \mleft( \psi'^{BC} \mright\| \mleft. \psi'^{B} \tensor \sigma^{C} \mright)
	 \right) + \log \frac{1}{\epsilon^2} \enspace.
	\]
\end{theorem}

For a complete proof of this result, including the correctness and error analysis of the protocol, see the proof of Theorem~1 in ref.~\cite{AJW18-Achievability-QSR}.

\subsection{Decoupling classical-quantum states}
\label{sec:decoupl-CQ}

\emph{Embezzlement\/} refers to a process introduced by van Dam and Hayden~\cite{vDH03-embezzling} in which any bipartite quantum state, possibly entangled, can be approximately produced from a bipartite catalyst using only local unitary operations. The bipartite catalyst is called the \emph{embezzling quantum state\/}. For an integer~$n$ and registers~$D$ and~$D'$ with~$|D| =|D'| \ge n$, the embezzling state is defined as
\begin{equation}
\label{eq:vDH-embezzlement}
    \ket{\xi}^{DD'} \quad \coloneqq \quad \frac{1}{\sqrt{S(n)}} 
    \sum_{i=1}^{n} \frac{1}{\sqrt{i}} \ket{i}^{D} \ket{i}^{D'} \enspace,
\end{equation}
where~$S(n) \coloneqq \sum_{i=1}^{n} \tfrac{1}{i} \;$. 
Van Dam and Hayden~\cite{vDH03-embezzling} showed that an arbitrary bipartite state can be embezzled from $\ket{\xi}^{DD'}$ with arbitrary accuracy when~$n$ is chosen to be correspondingly large. 
\begin{theorem}[\cite{vDH03-embezzling}]
Let ~$\ket{\phi}^{AB} \in \cH^{AB}$ be a bipartite state with Schmidt rank~$m$ and~$\ket{\xi}^{DD'}$ be the state defined in Eq.~\eqref{eq:vDH-embezzlement}. For~$\delta \in (0,1]$, there exists local isometries~$V_{1}:\cH^{D} \rightarrow \cH^{DA}$ and~$V_{2}: \cH^{D'} \rightarrow \cH^{D'B}$ such that
\begin{equation}
\label{eq:vDH-embezzlement-error}
    \rP \br{ (V_1 \tensor V_2) \ket{\xi} , \; \ket{\xi} \tensor \ket{\phi}} 
    \quad \le \quad \delta \enspace,
\end{equation}
provided that~$n\ge m^{2/\delta^2}$.
\end{theorem}

For a fixed~$a \in [n]$,  a close variant of the above embezzling state is defined as
\begin{equation} 
\label{eq:embezzling-state}
    \ket{\xi_{a:n}}^{DD'} \quad \coloneqq \quad \frac{1}{\sqrt{S(a,n)}} 
    \sum_{i=a}^{n} \frac{1}{\sqrt{i}} \ket{i}^{D} \ket{i}^{D'} \enspace,
\end{equation}
where~$S(a,n) \coloneqq \sum_{i=a}^{n} \tfrac{1}{i} \;$.
Using these states,  Lemma~\ref{lemma-unif-embz} below shows how we may embezzle the uniform distribution with closeness guaranteed in terms of max-relative entropy. The proof of Eq.~\eqref{eq-unif-embz-1}  in this lemma is due to Anshu and 
Jain~\cite[Claim 1]{AJ18-efficient-convex-split}, and Eq.~\eqref{eq-unif-embz-2}
follows from a similar argument. For completeness, we provide a proof for the lemma.

\begin{lemma} [Extension of \cite{AJ18-efficient-convex-split}, Claim 1]
\label{lemma-unif-embz}
	Let~$\delta \in (0, \frac{1}{15})$, and~$a, b, n \in \integers$ be positive integers
	such that~$a \geq b \ge 2$ and~$n \geq ~ a^{1/\delta}$. Let~$D$ and~$E$ be 
	registers with~$|D| \geq n$ and~$|E| \geq b$. Let~$W_b$ be a unitary operation that acts as
	\begin{equation} 
	\label{eq-unif-embz-unitary}
    	W_b \ket{i}^{D} \ket{0}^{E} 
    	\quad = \quad \ket{\floor{i/b}}^{D} \ket{i \tpmod b}^{E} \qquad \forall i\in \{0, \ldots |D|-1\}\enspace, 
	\end{equation}
	and~$\Pi_b \in \Pos \br{\cH^{DE}}$ be the projection operator onto the support of~$W_b \left( \xi_{a:n}^{D} \tensor \density{0}^{E} \right) W_b^\dagger$.
	It holds that
	\begin{equation}
	\label{eq-unif-embz-1}
	 W_b \left( \xi_{a:n}^{D} \tensor \density{0}^{E} \right) W_b^\dagger 
	 \quad \preceq \quad (1+15 \delta) ~ \xi_{1:n}^{D} \tensor \mu_{b}^{E} 
	 \enspace,
	 \end{equation}
	 and
	\begin{equation}
	\label{eq-unif-embz-2}
	\Pi_{b} \br{\xi_{1:n}^{D} \tensor \mu_{b}^{E}} \Pi_{b} \quad \preceq\quad 2 \; W_b \left( \xi_{a:n}^{D} \tensor \density{0}^{E} \right) W_b^\dagger \enspace.
	\end{equation}
	where~$\mu_{b}^{E} \eqdef \tfrac{1}{b} \sum_{e=0}^{b-1}\density{e}$.
\end{lemma}

\begin{proof}
    Let~$W_b$ be a unitary operator satisfying Eq.~\eqref{eq-unif-embz-unitary}. We have
    \begin{IEEEeqnarray}{rCl}
        W_b \left( \xi_{a:n}^{D} \tensor \density{0}^{E} \right) W_b^\dagger \quad
        & = & \quad \frac{1}{S(a,n)} \sum_{i=1}^{n} \frac{1}{i} \, W_b \left( \density{i}^{D} \tensor \density{0}^{E} \right) W_b^\dagger \nonumber \\
        & = & \quad \frac{1}{S(a,n)} \sum_{i=1}^{n} \frac{1}{i} \, \density{\floor{i/b}}^{D} \tensor \density{i \tpmod b}^{E} \nonumber \\
        & = & \quad \frac{1}{S(a,n)} \sum_{i' = \floor{\frac{a}{b}}}^{\floor{\frac{n}{b}}} 
              \sum_{e=0}^{\min \{ b-1 , n - i'b\}}  \frac{1}{bi'+e} \, \density{i'}^{D} \tensor \density{e}^{E}  \label{eq-unif-embz-prf-1}\\
        & \preceq& \quad \frac{1}{S(a,n)} \sum_{i' = \floor{\frac{a}{b}}}^{\floor{\frac{n}{b}}} 
              \sum_{e=0}^{ b-1}  \frac{1}{bi'} \, \density{i'}^{D} \tensor \density{e}^{E} \nonumber \\
        &  \preceq & \quad \frac{S(1,n)}{S(a,n)} \, \xi_{1:n}^{D} \tensor \mu_{b}^{E} \enspace. \label{eq-unif-embz-prf-2}
    \end{IEEEeqnarray}
    In ref.~\cite{Masch1790-adnotationes}, it is shown that~$\abs{S(a,n) - \log\tfrac{n}{a}} \leq 4$. Since~$n \geq a^{1/\delta}$, we have
    \begin{equation}
        \label{eq-unif-embz-prf-3}
           \frac{S(1,n)}{S(a,n)} \quad \le \quad \frac{\log n + 4}{\log n - \log a - 4} \quad \le \quad \frac{1+4\delta}{1-5\delta} \quad \le \quad 1+15\delta \enspace.
    \end{equation}
      Now, Eq.~\eqref{eq-unif-embz-prf-2} and Eq.~\eqref{eq-unif-embz-prf-3} together imply Eq.~\eqref{eq-unif-embz-1}. It remains to prove Eq.~\eqref{eq-unif-embz-2}. Let~$\Pi_b \in \Pos \br{\cH^{DE}}$ be the projection operator onto the support of~$W_b \left( \xi_{a:n}^{D} \tensor \density{0}^{E} \right) W_b^\dagger$.
      Eq.~\eqref{eq-unif-embz-prf-1} implies that
      \begin{equation*}
          \Pi_b \quad = \quad \sum_{i' = \floor{\frac{a}{b}}}^{\floor{\frac{n}{b}}} 
              \sum_{e=0}^{\min \{ b-1 , n - i'b\}}  \density{i'}^{D} \tensor \density{e}^{E} \enspace.
      \end{equation*}
      Thus,
      \begin{IEEEeqnarray*}{rCl's}
      \Pi_b \br{\xi_{1:n}^{D} \tensor \mu_{b}^{E} } \Pi_b \quad
      & = & \quad \frac{1}{S(1,n)} \, 
         \sum_{i' =  \floor{\frac{a}{b}}}^{\floor{\frac{n}{b}}} 
         \sum_{e=0}^{\min \{ b-1 , n - i'b\}}  \frac{1}{bi'} \, \density{i'}^{D} \tensor \density{e}^{E} \\
      & \preceq& \quad \frac{1}{S(1,n)} \, 
         \sum_{i' =  \floor{\frac{a}{b}}}^{\floor{\frac{n}{b}}} 
         \sum_{e=0}^{\min \{ b-1 , n - i'b\}}  \frac{2}{bi'+ e} \, \density{i'}^{D} \tensor \density{e}^{E} \\
      & = & \quad  \frac{ 2 \; S(a,n)}{S(1,n)} \, W_b \left( \xi_{a:n}^{D} \tensor \density{0}^{E} \right) W_b^\dagger \qquad \qquad \br{\text{by Eq.~\eqref{eq-unif-embz-1}}}\\
      & \preceq& \quad  2 \; W_b \left( \xi_{a:n}^{D} \tensor \density{0}^{E} \right) W_b^\dagger \enspace,
      \end{IEEEeqnarray*}
      where the first inequality holds since~$bi'+e \le 2 \, bi'$ for~$i' \ge 1$ and~$0 \le e \le b-1$, and the second inequality holds since~$S(a,n) \le S(1,n)$. 
\end{proof}

As a corollary of the above lemma, Anshu and Jain~\cite{AJ18-efficient-convex-split} show that the embezzling state~$\xi_{a:n}^{D}$ can be used almost catalytically to \emph{flatten\/} any quantum state using unitary operations. The proof of Eq.~\eqref{eq-flat-embz-1} in the corollary is provided in ref.~\cite[Eq.~(6)]{AJ18-efficient-convex-split}, and Eq.~\eqref{eq-flat-embz-2}
follows from Eq.~\eqref{eq-unif-embz-2}. For completeness, we provide a proof below.
\begin{corollary}[extension of \cite{AJ18-efficient-convex-split}, Eq.~(6)]
\label{cor-flatten}
	Let~$\rho \in \qstate( \cH^{C})$ be a quantum state with spectral decomposition~$\rho^{C} = \sum_c q(c) 
	\density{v_c}^{C}$. Let~$\delta \in (0,\tfrac{1}{15})$ and~$\gamma \in 
	(0,1)$ such that~$\tfrac{|C|}{\gamma}$ is an integer and all 
	eigenvalues~$q(c)$ are integer multiples of~$\tfrac{\gamma}{|C|}$. 
	Let~$a \coloneqq \tfrac{|C|}{\gamma} \max_{c} q(c)$,~$n \eqdef
	a^{1/\delta}$, and~$D$ and~$E$ be quantum registers with~$|D| \geq n$ 
	and~$|E|=a$. 
	Let~$W \in \unitary(\cH^{CED})$ be the unitary operator defined as
	\[
	W \quad \coloneqq \quad \sum_c \density{v_c}^C \tensor 
	W_{b(c)}^{ED} \enspace
	\]
	and~$\Pi \in \Pos(\cH^{CED})$ be the projection operator defined as
	\begin{equation*}
	    \Pi \quad \coloneqq \quad \sum_c \density{v_c}^C \tensor 
	\Pi_{b(c)}^{ED} \enspace,
	\end{equation*}
	where~$W_{b(c)}$ and~$\Pi_{b(c)}$ are the operators defined in 
	Lemma~\ref{lemma-unif-embz} with~$b(c) \coloneqq \tfrac{q(c)|C|}{\gamma}$ (but with the tensor factors corresponding to~$D$ and~$E$ swapped). 
	Then, we have
	\begin{equation}
	\label{eq-flat-embz-1}
	     W \left( \rho^{C} \tensor \density{0}^{E}  \tensor \xi_{a:n}^{D} \right) W^\dagger 
	     \quad \preceq \quad (1+15 \delta) ~ \rho^{CE} \tensor \xi_{1:n}^{D}
	\end{equation}
	and 
	\begin{equation}
	\label{eq-flat-embz-2}
	    \Pi \br{\rho^{CE} \tensor \xi_{1:n}^{D}} \Pi \quad \preceq\quad 2 \; W \left( \rho^{C} \tensor \density{0}^{E}  \tensor \xi_{a:n}^{D} \right) W^\dagger \enspace,
 	\end{equation}
	where~$\rho^{CE} \coloneqq \frac{\gamma}{|C|} \sum_{c} \density{v_c}^C \tensor
	\sum_{e=0}^{b(c)-1} \density{e}^E$ is an extension 
	of~$\rho^{C}$ with flat spectrum. 
\end{corollary}

\begin{proof}
    Let~$W$ be the unitary operator defined in the statement of the corollary . We have
    \begin{IEEEeqnarray*}{rCl}
    \IEEEeqnarraymulticol{3}{l}{
    W \left( \rho^{C} \tensor \density{0}^{E}  \tensor \xi_{a:n}^{D} \right) W^\dagger} \\
    & = & \quad \sum_c q(c) \density{v_c}^{C} \tensor W_{b(c)} \br{\density{0}^{E}  \tensor \xi_{a:n}^{D}} W_{b(c)}^\dagger \\
    & \preceq& \quad \br{1+15\delta} ~ \sum_c q(c) \density{v_c}^{C} \tensor \frac{\gamma}{q(c)|C|} \sum_{e=0}^{b(c)-1}\density{e}^{E}  \tensor \xi_{a:n}^{D}   \\
    & = & \quad \br{1+15\delta} ~ \rho^{CE} \tensor \xi_{a:n}^{D} \enspace,
    \end{IEEEeqnarray*}
    where the inequality follows from Lemma~\ref{lemma-unif-embz}. So, it remains to prove Eq.~\eqref{eq-flat-embz-2}. Let~$\Pi$ be the projection operator defined in the statement of the corollary. We have
    \begin{IEEEeqnarray*}{rCl}
    \Pi \left( \rho^{CE} \tensor \xi_{1:n}^{D} \right) \Pi \quad 
    & = & \quad \frac{\gamma}{|C|} \sum_c b(c) \density{v_c}^{C} \tensor \Pi_{b(c)} \br{\mu_{b(c)}^{E} \tensor \xi_{a:n}^{D}} \Pi_{b(c)}  \\
    & \preceq& \quad 2 \sum_c q(c) \density{v_c}^{C} \tensor W_{b(c)}\br{\density{0}^{E}  \tensor \xi_{a:n}^{D}}W_{b(c)}^\dagger   \\
    & = & \quad 2 \; W \br{\rho^{C}\tensor \density{0}^{E} \tensor \xi_{a:n}^{D}} W^\dagger \enspace,
    \end{IEEEeqnarray*}
    where the inequality is a consequence of  Lemma~\ref{lemma-unif-embz}. 
\end{proof}

We use the above flattening procedure to decouple the quantum register in a classical-quantum state.  
\begin{corollary}
\label{cor-decoupling-C-Q-state}
Consider a classical-quantum state~$\rho^{JC} \eqdef \sum_j p(j) \; \density{j}^J 
     \tensor \rho_j^{C}$, where ~$p$ is a probability distribution and~$\rho_j^{C} \in \qstate \! \left( \cH^{C}\right)$.  Let~$\delta \in (0,\tfrac{1}{15})$ 
     and~$\gamma \in (0,1)$ such that~$a \coloneqq \tfrac{|C|}{\gamma}$ is an integer 
     and suppose that the eigenvalues of all the states~$\rho_j^{C}$ are integer multiples of~$\tfrac{\gamma}{|C|} \;$.
     Let~$n \eqdef a^{1/\delta}$,~$D$ and~$E$ be quantum registers 
     with~$|D| \geq n$ and~$|E|=a$. Then, there 
     exists a unitary operator~$U \in \unitary(\cH^{JCED})$, read-only on 
     register~$J$, and a projection operator~$\Pitilde \in \Pos(\cH^{JCED})$ such that
\begin{align}
\label{eq:cor-decoupl-1}
U \left( \rho^{JC} \tensor \density{0}^{E}  \tensor \xi_{a:n}^{D} \right) U^\dagger 
    & \quad \preceq \quad (1+15 \delta) ~ \rho^{J} \tensor \nu^{CE} \tensor 
        \xi_{1:n}^{D}  \enspace, \\
\label{eq:cor-decoupl-2}
\Pitilde \left( \rho^{J} \tensor \nu^{CE} \tensor \xi_{1:n}^{D} \right) \Pitilde 
    & \quad \preceq \quad 2 \; U \left( \rho^{JC} \tensor \density{0}^{E}  \tensor 
        \xi_{a:n}^{D} \right) U^\dagger  \enspace,
\end{align}
and
\begin{equation}
     \label{eq:cor-decoupl-3}
         \trace \left[ \Pitilde U \left( \rho^{JC} \tensor \density{0}^{E}  \tensor \xi_{a:n}^{D} 
	  \right) U^\dagger  \right]
         \quad = \quad 1 \enspace,
\end{equation}
where~$\nu^{CE} \eqdef \frac{1}{a} \sum_{s=0}^{a-1} \density{s}^{CE}$.
\end{corollary}
\begin{proof} 
Notice that the integers~$a$ and~$n$ and registers~$D$ and~$E$ satisfy the properties required in Corollary~\ref{cor-flatten}. For each~$j$, let~$W^{(j)}$ be the unitary operator given by Corollary~\ref{cor-flatten} for flattening~$\rho_j^{C} \eqdef \sum_{c} q_j(c) \; \density{v_{j,c}}$. Hence, we can flatten all~$\rho_{j}^{C}$ simultaneously using the unitary operator~$U_1 \eqdef \sum_{j} \density{j} \tensor W^{(j)}$, and we get
\begin{equation*}
    U_{1} \left( \rho^{JC} \tensor \density{0}^{E} \tensor \xi_{a:n}^{D} \right) U_1^{\dagger} \quad \preceq\quad (1+15 \delta) \sum_{j} p(j) \; \density{j}^{J} \tensor \rho_{j}^{CE} \tensor \xi_{1:n}^{D}  \enspace,
\end{equation*}
where~$\rho_j^{CE} \eqdef  \frac{\gamma}{|C|} \sum_{c} \density{v_{j,c}}^C \tensor \sum_{e=0}^{q_j(c) |C|/\gamma} \density{e}^E$ is an extension 
of~$\rho^{C}$ with flat (i.e., uniform) spectrum. For each~$j$, the support of~$\rho_j^{C}$ has dimension~$\sum_{c} q_j(c) \frac{|C|}{\gamma}$, which equals~$a$ independent of~$j$. Hence, there exists a unitary operator~$V^{(j)}$ mapping~$\rho_{j}^{CE}$ to~$\nu^{CE}$.  Let~$U_2 \in \unitary(\cH^{JCE})$ be the unitary operator~$U_2 \coloneqq \sum_{j} \density{j} \tensor V^{(j)}$. Then, the unitary operator~$U \coloneqq U_2 U_1$ satisfies Eq.~\eqref{eq:cor-decoupl-1}.

Now, for each~$j$, let~$\Pi^{(j)} \in \Pos(\cH^{CED})$ be the projection operator given by Corollary~\ref{cor-flatten}. Define~$\Pi' \coloneqq \sum_j \density{j} \tensor \Pi^{(j)} $ and~$
\Pitilde \coloneqq  U_2 \Pi' U_2^\dagger$.
We have
\begin{IEEEeqnarray*}{rCl}
    \Pitilde \left( \rho^{J} \tensor \nu^{CE} \tensor 
      \xi_{1:n}^{D} \right) \Pitilde \quad 
      & = & \quad  U_2 \Pi' U_2^\dagger  \left( \rho^{J} \tensor \nu^{CE} \tensor \xi_{1:n}^{D} \right)  U_2 \Pi' U_2^\dagger \\
      & = & \quad U_2 \Pi' \br{\sum_{j} p(j) \; \density{j}^{J} \tensor \rho_{j}^{CE} \tensor \xi_{1:n}^{D} }  \Pi' U_2^\dagger \\
      & = &  \quad U_2 \br{\sum_{j} p(j) \; \density{j}^{J} \tensor \Pi^{(j)} \br{\rho_{j}^{CE} \tensor \xi_{1:n}^{D} } \Pi^{(j)}} U_2^\dagger \\
      & \preceq& \quad 2 \; U_2 \br{\sum_{j} p(j) \; \density{j}^{J} \tensor W^{(j)} \br{\rho_{j}^{C} \tensor \density{0}^{E} \tensor \xi_{a:n}^{D}} {W^{(j)}}^{\dagger}}  U_2^\dagger \\
      & = & \quad 2 \; U_2 U_1 \br{\sum_{j} p(j) \; \density{j}^{J} \tensor \rho_{j}^{C} \tensor \density{0}^{E} \tensor \xi_{a:n}^{D}}  U_1^\dagger U_2^\dagger \\
      & = & \quad 2 \; U \left( \rho^{JC} \tensor \density{0}^{E} \tensor \xi_{a:n}^{D} \right) U^\dagger \enspace,
\end{IEEEeqnarray*}
where the inequality follows from Corollary~\ref{cor-flatten}, Eq.~\eqref{eq-flat-embz-2}.

Moreover, by the construction in Lemma~\ref{lemma-unif-embz} and Corollary~\ref{cor-flatten}, for each~$j$, the operator~$\Pi^{(j)}$ is the projection operator onto the support of~$W^{(j)}\br{\rho_{j}^{C} \tensor \density{0}^{E} \tensor \xi_{a:n}^{D}} {W^{(j)}}^\dagger$. Hence, we have
\begin{IEEEeqnarray*}{rCl}
 \trace \left[ \Pitilde U \left( \rho^{JC} \tensor \density{0}^{E}  \tensor \xi_{a:n}^{D} 
	  \right) U^\dagger \right] \quad 
 & = & \quad \trace \left[ \Pi' U_1 \left( \rho^{JC} \tensor \density{0}^{E}  \tensor \xi_{a:n}^{D} 
	  \right) U_1^\dagger \right] \\
 & = & \quad \sum_{j} p(j) \trace \left[ \Pi^{(j)} W^{(j)}\br{\rho_{j}^{C} \tensor \density{0}^{E} \tensor \xi_{a:n}^{D}} {W^{(j)}}^\dagger \right] \\
 & = & \quad 1 \enspace.
\end{IEEEeqnarray*}
This completes the proof.
\end{proof}

\begin{remark}
In the above corollary, we assume that the eigenvalues of~$\rho^{C}_j$ are rational. We can approximate an arbitrary state with one that has only rational eigenvalues with arbitrary accuracy, since the set of rational numbers is dense in the set of reals. Consequently, the error with respect to the max-relative entropy can also be made arbitrarily close to zero.
\end{remark}

\section{The new protocol}
\label{sec-new-protocol}

In this section, we present and analyse the new protocol for one-shot state redistribution. This proves the main result in this article, as stated more precisely in the following theorem.

\begin{theorem}
 \label{thm-main}
 	Let~$\ket{\psi}^{RABC}$ be a pure quantum state shared between 
 	a referee~$(R)$, Alice~$(AC)$ and Bob~$(B)$. 
 	For every~$\epsilon_1,\epsilon_2 \in (0, 1)$ satisfying~$\epsilon_1+9 \epsilon_2 \le 1$, there exists 
 	an entanglement-assisted one-way protocol operated by Alice and Bob which 
 	starts in the state~$\ket{\psi}^{RABC}$, and outputs a state~$\phi^{RABC} 
 	\in \sB^{\epsilon_1 + 9 \epsilon_2}(\psi^{RABC})$ where registers~$A$,~$BC$, and~$R$ are held by Alice, Bob and Ref, respectively.
 	The communication cost of this protocol is bounded from above by
 	\begin{equation}
 	\label{eq:main-cost}
        \frac{1}{2}  \inf_{\psi' \in \sB^{\epsilon_1}(\psi^{RBC})} \; 
        \inf_{\sigma \in \ME^{\epsilon_2^4/4,\psi'}} \left[\Dmax\mleft({\psi'}^{RBC} 
        \Big\| ~ \sigma^{RBC} \mright) - \Dh^{\epsilon_2^2} 
        \mleft({\psi'}^{BC} \Big\| ~ \sigma^{BC} \mright) \right]
        + \log \frac{1}{\epsilon_2^{2}} + 1\enspace.
 	 \end{equation}
 \end{theorem}
We get Theorem~\ref{thm:oneshotcondmutrelent} by choosing~$\epsilon_2^{2} = \epsilon_1 = \epsilon$.

We describe a protocol for redistributing~$\ket{\psi}^{RABC}$ with error~$9\epsilon_2$ and cost at most
 \begin{equation}
 	\label{eq:nonsmooth-cost}
     \frac{1}{2} \min_{\sigma^{RBC}\in \ME^{\epsilon_2^4/4,\psi}} \mleft[ \Dmax \br{\psi^{RBC} \big\| \sigma^{RBC}} 
    -  \Dh^{\epsilon_2^{2}}\br{\psi^{BC} \big\|\sigma^{BC}} \mright]+  \log\frac{1}{\epsilon_2^2} + 1 \enspace.
\end{equation}
Then, Theorem~\ref{thm-main} follows since for every~$\ket{\psi'}\in \sB^{\epsilon_1}(\ket{\psi}^{RABC})$, Alice and Bob can assume that the global state is~$\ket{\psi'}^{RABC}$, and run the protocol for~$\ket{\psi'}$. This protocol redistributes the state~$\ket{\psi}$ with  additional error at most~$\epsilon_1$. 

Let~$\sigma^{RBC}$ be a quantum Markov extension of~$\psi^{RB}$. If~$\sigma^{RBC} = {\psi}^{RB} \tensor {\psi}^{C}$, Alice and Bob can redistribute~$\psi^{RABC}$ with error~$9 \epsilon_2 > 0$ and communication cost bounded by Eq.~\eqref{eq:nonsmooth-cost} using the AJW protocol.
However, in general,~$\sigma^{RBC}$ is not necessarily a product state. In that case, we design a reduction procedure which allows us to use the AJW protocol as a subroutine. This procedure decouples~$C$ from~$RB$ when applied to~$\sigma^{RBC}$, while preserving~$\psi^{RB}$ when applied to~$\psi^{RBC}$. This procedure is  similar to the \emph{conditional erasure\/} task in Refs.~\cite{BBMW18-cond-decoup, BBMW18-Deconst-cond-erase} except that, here, the decoupling and negligible disturbance properties are desired for two possibly different quantum states.
    
In the rest of this section, we first explain a simplified version of the reduction procedure and the protocol for the special case that register~$A$ is trivial and~$\ket{\psi}^{RBC}$ is the GHZ state. This illustrates the key components underlying the reduction. Then, in Section~\ref{sec:protocol-general},  we provide the complete version of the reduction procedure and the protocol for redistributing an arbitrary quantum state~$\ket{\psi}^{RABC}$.

\subsection{The GHZ state example} 
\label{sec-GHZ}

To elaborate on the reduction procedure, we start with the example where~$\psi^{RBC}$ is the GHZ state
\begin{equation*}
  \frac{1}{\sqrt{d}} \sum_{j=1}^d \ket{j}^R \ket{j}^B \ket{j}^C \enspace,  
\end{equation*}
and the Markov extension~$\sigma^{RBC}$ of~$\psi^{RB}$ is
\begin{equation*}
\label{eq:QME-GHZ}
   \frac{1}{d} \sum_{j=1}^{d} \density{j}^R \tensor \density{j}^B \tensor \density{j}^C \enspace. 
\end{equation*}
The reduction broadly follows the description we gave in Section~\ref{sec-intro-techniques}, and is a two-step process. We expand on these steps below. 

\textbf{(1) Coherent measurement of register~$B$.}
By ``coherent measurement'', we mean the application of the isometry given by a Steinspring representation of the measurement. For the GHZ state, this corresponds ``copying'' the content of register~$B$ into a fresh register, in superposition. The state of the fresh register is chosen so as to facilitate the redistribution protocol.
Let~$T$ be a register with~$|T|=d$, and~$\ket{\Psi}^{TT'} \eqdef \tfrac{1}{ \sqrt{d}} \sum_t \ket{tt}$ be the maximally entangled state over registers~$T$ and~$T'$.
\suppress{
and~$\set{\ket{t}}_{t=0}^{d-1}$ be a basis for~$\cH^{T}$. For~$a,b \in \set{0, \ldots, d-1}$, let~$P_{a,b} \in \unitary \br{\cH^{T}}$ be the \emph{Heisenberg-Weyl\/} operator defined as~$P_{a,b} \coloneqq \sum_{t} \exp\br{{\tfrac{2\pi i tb}{d}}} \ketbra{t+a}{t}^T$
}
Define the unitary operator~$U_1 \in \unitary ( \cH^{BT} )$ as~$U_1 \coloneqq \sum_{j} \density{j}^{B} \tensor P_{j}^{T}$, where~$P_j$ is the Heisenberg-Weyl operator as defined in Section~\ref{sec-background}. Let~$\ket{\kappa_1}^{RBCTT'}$ and~$\tau_1^{RBCT}$ be the states obtained by applying~$U_1$ to~$\ket{\psi}^{RABC} \tensor \ket{\Psi}^{TT'}$ and~$\sigma^{RBC} \tensor \Psi^{T}$, respectively. We have
\begin{align*}
\ket{\kappa_1}^{RBCTT'} \quad & = \quad \frac{1}{d} \sum_{j=1}^d \ket{j}^R \ket{j}^B \ket{j}^C \tensor \sum_{t = 1}^d \ket{t \oplus j}^T \ket{t}^{T'} \enspace.
\end{align*}
Since the set of Heisenberg-Weyl operators~$\set{P_a}$ is closed under multiplication, and each~$P_{a}$ is traceless unless~$a=d$, the states~$(P_a \tensor \id) \ket{\Psi}$ are mutually orthogonal. So the unitary operator~$U_1$ coherently measures register~$B$ in~$\psi^{RBC}$ while it acts trivially on~$\sigma$. Moreover, the reduced state on~$T$ remains maximally mixed. So
\begin{align*}
\kappa_1^{RBC} \quad & = \quad \frac{1}{d} \quad \sum_{j} \density{j}^R \tensor \density{j}^B \tensor \density{j}^C \enspace, \qquad \text{and} \\
\tau_1^{RBCT} \quad & = \quad \sigma^{RBC} \tensor \frac{\id^{T}}{d} \enspace.
\end{align*}
 	
\textbf{(2) Decoupling~$C$ from~$RB$ in~$\sigma$.}
Let~$U_2 \in \unitary(\cH^{BC})$ be a unitary operator that is read-only on~$B$ 
and maps~$\ket{j}^C$ to~$\ket{0}^{C}$ if system~$B$ is in the 
state~$\ket{j}$. Let~$\ket{\kappa_2}^{RBCTT'}$ and~$\tau_2^{RBCT}$ 
be the states after applying~$U_2$ to~$\ket{\kappa_1}^{RBCTT'}$ 
and~$\tau_1^{RBCT}$, respectively. We have
\begin{align*}
\ket{\kappa_2}^{RBCTT'} \quad & = \quad \frac{1}{d} \quad \sum_{j} \ket{j}^R \tensor \ket{j}^B \tensor \ket{0}^C \tensor \sum_{t = 1}^d \ket{t \oplus j}^T \ket{t}^{T'} \enspace, \qquad \text{and} \\
    \tau_2^{RBCT} \quad & = \quad \psi^{RB} \tensor \density{0}^{C} \tensor \frac{\id^{T}}{d} \enspace.
\end{align*}
In particular, since register~$B$ is classical in~$\kappa_1^{RBC}$ and~$U_2$ is read-only on~$B$, we get~$\kappa_2^{RB} = \psi^{RB}$.

The reduction procedure uses the above two steps to (effectively) add the maximally mixed state~$\Psi^{T}$ and apply the unitary operator~$U_2 U_1$. Note that running this procedure on both~$\psi$ and~$\sigma$ does not change their max-relative entropy and the hypothesis testing entropy. We have
\begin{multline}
\label{eq:GHZ-dist}
    \Dmax \br{\psi^{RBC}\|\sigma^{RBC}} -  \Dh^{\epsilon_2^{2}}\br{\psi^{BC}\|\sigma^{BC}} \quad
    = \quad \Dmax \br{\kappa_2^{RBCT}\|\tau_2^{RBCT}} -  \Dh^{\epsilon_2^{2}}\br{\kappa_2^{BCT}\|\tau_2^{BCT}} 
\end{multline}
where~$\tau_2^{RBCT} = \kappa_2^{RB} \tensor \density{0}^{C} \tensor \tfrac{\id^{T}}{d}$. Hence, if Alice and Bob locally map~$\ket{\psi}$ to~$\ket{\kappa_2}$, then they can run the AJW protocol to transfer registers~$CT$ to Bob and finally retrieve~$\ket{\psi}$ by applying~$U_1^{-1}U_2^{-1}$. A hitch here is that the reduction procedure cannot be implemented directly (i.e., as described above) for the local transformation of~$\ket{\psi}$ to~$\ket{\kappa_2}$. This is because register~$C$ is initially with Alice and 
register~$B$ is with Bob. However, since~$\psi^{RB} = \kappa_2^{RB}$, there is an isometry~$V: \cH^{AC} \rightarrow \cH^{ACTT'}$ which maps~$\ket{\psi}^{RABC}$ to~$\ket{\kappa_2}^{RABCTT'}$, as guaranteed by the Uhlmann theorem. \emph{Alice\/} can thus implement the local transformation from~$\ket{\psi}$ to~$\ket{\kappa_2}$.

In summary, the simplified version of the protocol for the GHZ state works as follows:
\begin{enumerate}
  \item Alice applies the isometry~$V$ on her registers~$AC$, and transforms the global state to the state~$\ket{\kappa_2}^{RABCTT'}$ such that registers~$(ACTT')$,~$(B)$, and~$(R)$ are with Alice, Bob and Ref, respectively.
  
  \item Choosing~$\sigma^{CT} \eqdef \density{0}^{C} \tensor \tfrac{\id^{T}}{d}$, Alice and Bob run the AJW protocol on~$\ket{\kappa_2}$ to transfer registers~$CT$ to Bob with error at most~$9 \epsilon_2$. Let~$\widehat{\kappa}_2^{RABCTT'}$ be the joint state of the registers~$RABCTT'$ at the end of this step.
  
  \item Bob applies~$U_1^{-1}U_2^{-1}$ on the registers~$BCT$, which are now in his possession.
  
  \item The output of the protocol is the final state in registers~$RABC$.
\end{enumerate}
By Theorem~\ref{thm-QSR-AJW} and Eq.~\eqref{eq:GHZ-dist}, the cost of the above protocol is at most
\[
    \Dmax \br{\psi^{RBC}\|\sigma^{RBC}} -  \Dh^{\epsilon_2^{2}}\br{\psi^{BC}\|\sigma^{BC}} + \log \frac{1}{\epsilon_2^{2}} \enspace,
\]
and~$\rP(\kappa_2^{RABCTT'}, \widehat{\kappa}_2^{RABCTT'}) \le 9 \epsilon_2 \;$.
Let~$\phi^{RABC}$ be the final state of the registers~$RABC$. We have
\begin{IEEEeqnarray*}{rCl}
    \rP \br{\psi^{RABC}, \phi^{RABC}} \quad 
    & \le & \quad \rP \br{\psi^{RABC} \tensor \Psi^{TT'}, \phi^{RABCTT'}} \\
    & = &  \quad \rP \br{\kappa_2^{RABCTT'}, \widehat{\kappa}_2^{RABCTT'}} \\
    & \le & \quad 9 \epsilon_2 \enspace,
\end{IEEEeqnarray*}
where the first inequality is obtained by considering extensions of states in~$RABC$ to those in~$RABCTT'$ and the monotonicity of purified distance under quantum operations, and the second step follows by the invariance of purified distance under unitary operations (in this case~$U_2 U_1$).

\subsection{The protocol for arbitrary states} 
\label{sec:protocol-general}

Now consider an arbitrary state~$\ket{\psi}^{RABC}$ and a quantum Markov extension~$\sigma^{RBC} \in \ME^{\epsilon_2^4/4,\psi}$. As explained in Section~\ref{sec:QMarkov}, there exists a decomposition of register~$B$ as~$\cH^{B}  =  \bigoplus_{j} \cH^{B_{j}^{R}} \tensor \cH^{B_{j}^{C}}$ such that 
\begin{equation}
        {\psi}^{RB} \quad = \quad \sigma^{RB} \quad = \quad \bigoplus_{j} p(j) \, {\psi}_{j}^{RB_j^{R}} 
 	    \tensor {\psi}_{j}^{B_{j}^{C} } \enspace,
\end{equation}
and
\begin{equation}
    \label{eq:sigma-structure}
 	    \sigma^{RBC} \quad = \quad \bigoplus_{j} p(j) \, \sigma_{j}^{RB_j^{R}} 
 	    \tensor \sigma_{j}^{B_{j}^{C} C} \enspace,
\end{equation}
where~$\sigma_{j}^{RB_j^{R}} = {\psi}_{j}^{RB_j^{R}}$,~$\sigma_{j}^{B_{j}^{C} C} \in \sB^{\epsilon_2^4/4} \left(\trace_{B_{j}^{R}}\br{(\Pi_j \tensor \id){\psi}^{BC} (\Pi_j \tensor \id)} \right)$ and~$\Pi_j$ is the projection operator over the~$j$-th subspace in the direct sum decomposition of~$\cH^B$. This special structure of~$\sigma^{RBC}$ makes it possible to design the reduction procedure. As in the case of the GHZ state, the reduction procedure consists of the two main steps of coherent measurement and decoupling. These are preceded by two pre-processing steps. The pre-processing steps unitarily transform~$\psi$ and~$\sigma$ to the states~$\kappa$ and~$\tau$ which are easier to handle. In step~(i), we apply a local isometry transforming~$\sigma^{RBC}$ to a classical-quantum state.

\textbf{(i) Viewing~$\sigma^{RBC}$ as a classical-quantum state.} 
Let~$B^{R}$ and~$B^{C}$ be two quantum registers with~$\abs{B^{R}} \eqdef \max_j \abs{B_{j}^{R}}$
and~$\abs{B^{C}} \eqdef \max_j \abs{B_{j}^{C}}$. As a consequence of Eq.~\eqref{eq:sigma-structure},
there exists an isometry~$U_{\ri}: \cH^{B} \rightarrow 
 	\cH^{B^{R}JB^{C}}$ which takes~$\sigma^{RBC}$ to the 
 	state
 	\begin{equation}
 	\label{eq:sigma1-def}
 	    \tsigma^{RB^{R}JB^{C}C} \quad \coloneqq \quad \sum_j p(j) \,  \sigma_{j}^{RB^{R}} 
 	    \tensor \density{j}^J \tensor \sigma_{j}^{B^{C}C} \enspace.
 	\end{equation}
Let~$\ket{\psi_1}^{RAB^{R}JB^{C}C}$ be the state obtained by applying the same operation on~$\ket{\psi}^{RABC}$, i.e.,
\begin{align}
\label{eq-psi1}
 	\ket{\psi_1}^{RAB^{R}JB^{C}C} \quad & \coloneqq \quad U_{\ri} \ket{\psi}^{RABC}  \quad =  \quad \sum_{j,j'} \ketbra{j}{j'}^{J} \tensor \psi_{j,j'}^{RAB^{R}B^{C}C} \enspace,
\end{align}
for some sub-normalized, rank~$1$ states~$\psi_{j,j'}$.
It is sufficient to design a protocol for redistributing register~$C$ in~$\ket{\psi_1}^{RAB^RJB^C C}$ when initially registers~$(AC)$ are held by Alice,~$(B^R JB^C)$ are held by Bob and~$R$ is held by Ref.
Notice that~$\psi_1^{RB^R J B^C} =\sigma^{RB^R J B^C}$ since~$\psi^{RB} = \sigma^{RB}$. So~$\psi_1^{RB^R J B^C}$ is a quantum Markov state of the form~$RB^R \!-\! J \!-\! B^C$. So, Alice and Bob can use the folklore protocol for redistributing quantum Markov states explained in Fig.~\ref{fig:figure1} and transfer~$B^C$ to Alice. This is done in step~(ii) of pre-processing.
 	
\textbf{(ii) Transferring~$B^{C}$ from Bob to Alice without communication.} 
Note that~$\psi_1^{RB^{R}JB^{C}}$ is purified by systems~$(AC)$ which are with Alice. So by applying a suitable isometry, Alice can prepare the following purification of $\psi_1^{RB^{R}JB^{C}}$:
\[
\ket{\widehat{\psi}_1}^{RB^{R}JJ'B^{C}GH} \quad \coloneqq \quad \sum_j \sqrt{p(j)}  \ket{\sigma_{j}}^{RB^{R}G} 
 	    \tensor \ket{j,j}^{JJ'} \tensor \ket{\sigma_{j}}^{B^{C}H} \enspace,
 	   \]
 where registers~$J'GH$ are held by Alice. Let~$\delta_1\in (0,1)$,~$n_1 \coloneqq \abs{B^{C}H} ^{2/\delta_1^2}$, and~$D_1,D_1'$ be registers with~$|D_1|=|D_1'|=n_1$. Conditioned on  register~$J$, Alice and Bob use the embezzling state~$\ket{\xi}^{D_1D_1'}$ (as defined in Eq.~\eqref{eq:vDH-embezzlement}) and the reverse of the van Dam-Hayden protocol~\cite{vDH03-embezzling} to embezzle out~$\ket{\sigma_{j}}^{B^{C}H}$ in superposition. They thus obtain a state~$\widetilde{\psi}_1$ such that 
 \begin{equation*}
     \rP \mleft( \widetilde{\psi}_1^{R B^R GJJ'D_1 D_1'} , ~ \sum_j \sqrt{p(j)}  \ket{\sigma_{j}}^{RB^{R}G}\tensor \ket{j,j}^{JJ'}\otimes \ket{\xi}^{D_1D_1'} \mright) \quad \le \quad \delta_1 \enspace.
 \end{equation*}
 Finally, conditioned on register~$J$, Alice locally generates~$\ket{\sigma_{j}}^{B^{C}H}$ in superposition with registers~$B^CH$ on her side, and applies an Uhlmann unitary operator to her registers in order to prepare the purification~$\ket{\psi_1}^{RAB^{R}JB^{C}C}$. Let~$ U_{\rii,A}$ and~$U_{\rii,B}$ denote the overall unitary operators applied by Alice and Bob, respectively, in this step. After applying ~$ U_{\rii,A}$ and~$U_{\rii,B}$,  the global state is~$\ket{\psi_2}$ satisfying
 \begin{equation*}
     \rP \mleft( \psi_2^{RAB^{R}JB^{C}CD_1D_1'} , ~ \density{\psi_1}^{RAB^{R}JB^{C}C} \otimes \density{\xi}^{D_1D_1'} \mright) \quad \le \quad \delta_1 \enspace,
 \end{equation*}
where registers~$AB^{C}C$ are with Alice, registers~$B^{R}J$ are with Bob and register~$R$ is with Ref. 
Thus, the problem reduces, up to a purified distance~$\delta_1$, to the case where the global state is~$\ket{\psi_1}$ and the register~$B^C$ is with Alice. Henceforth, we assume that this is indeed the case. We account for the inaccuracy introduced by this assumption in the error analysis of the protocol. This completes the second step and the pre-processing stage of the protocol.

Due to the pre-processing steps, we may suppose that the global state is~$\ket{\psi_1}^{RAB^RJB^CC}$ such that 
registers~$(AB^C C)$,~$(B^RJ)$, and~$R$ are held by Alice, Bob, and Ref, 
respectively. It then remains for Alice to send~$B^C C$
to Bob. To achieve this, we follow a two-step unitary procedure (as in the case of the GHZ state) which
decouples registers~$RB^RJ$ and~$B^C C$ in~$\tsigma^{RB^{R}JB^{C}C}$ 
while keeping the state of registers~$RB^R J$ unchanged. This operation 
transforms~$\tsigma$ to a product state and allows us to use the AJW 
protocol as a subroutine to achieve the redistribution with the desired communication cost and accuracy.

To decouple~$R B^{R} J$ from~$B^C C$ in~$\tsigma$, we would like to use embezzlement and 
the unitary operator given by Corollary~\ref{cor-decoupling-C-Q-state}. This unitary operator
acts on registers~$J B^C C$ and is read-only on register~$J$. However, 
since register~$J$ is not necessarily classical in~$\psi_1^{RB^R J B^C C}$, the operation may disturb the marginal state~$\psi_1^{RB^R J}$. So as in the example of the GHZ state, we
resolve this issue by first coherently measuring register~$J$ using an additional 
maximally entangled state. This operation transforms~$\psi_1^{RB^RJB^C C}$ 
to a classical-quantum state, classical in register~$J$, and keeps~$\tsigma^{RB^RJB^C C}$ intact. The following two steps contain the detailed construction of these unitary procedures.

\textbf{(1) Coherent measurement of register~$J$.} 
 	Let~$F$ be a register with~$|F|=|J|$, and let~$d \eqdef \size{F}$. Let~$P_{j} \in \unitary \br{\cH^{F}}$ be a Heisenberg-Weyl operator as defined in Section~\ref{sec-background}. Let~$U_1 \in \unitary ( \cH^{JF} )$ be a unitary operator defined as~$U_1 \coloneqq \sum_{j} \density{j}^{J} \tensor P_{j}^{F}$. Define
 	\begin{equation*}
 	    \ket{\kappa_1}^{RAB^{R}JB^{C}CFF'} \quad \coloneqq \quad U_1 \left(\ket{\psi_1}^{RAB^{R}JB^{C}C} \tensor \ket{\Psi}^{FF'}\right) \enspace,
 	\end{equation*}
 	and
 	\begin{equation}
 	\label{eq-tau1}
 	    \tau_1^{RB^{R}JB^{C}CF} \quad \coloneqq \quad U_1 \left(\tsigma^{RB^{R}JB^{C}C} \tensor \frac{\id^{F}}{|F|} \right) U_1^{\dagger}\enspace,
 	\end{equation}
 	where~$\ket{\Psi}^{FF'} \eqdef \tfrac{1}{\sqrt{d}} \sum_{f = 1}^d \ket{ff}$ is the maximally entangled state over registers~$F$ and~$F'$. For the same reasons as in the GHZ example, the unitary operator~$U_1$ acts trivially on~$\tsigma$ while it measures register~$J$ in~$\psi_1^{RB^{R}JB^{C}C}$ coherently. In particular,
 	\begin{equation}
 	\label{eq:tau1}
 	    \tau_1^{RB^{R}JB^{C}CF} \quad = \quad \tsigma^{RB^{R}JB^{C}C} \tensor \frac{\id^{F}}{|F|} \enspace,
 	\end{equation}
 	and
 	\begin{equation}
 	    \kappa_1^{RB^{R}JB^{C}C} \quad = \quad \sum_{j} \density{j}^{J} \tensor \psi_{j,j}^{RB^{R}B^{C}C} \enspace.
 	\end{equation}
 	
\textbf{(2) Decoupling registers~$B^{C}C$ from~$RB^{R}J$ in~$\tau_1 \,$.} 
By Eqs.~\eqref{eq:sigma1-def} and~\eqref{eq:tau1}, register~$J$ is classical 
 	in~$\tau_1^{RB^{R}JB^{C}C}$ and conditioned on~$J$, registers~$RB^{R}$ are decoupled 
 	from~$B^{C}C$. Hence, we can decouple registers~$B^{C}C$ from registers~$RB^{R}J$ in~$\tau_1$ 
 	using embezzling states and applying the unitary operator given in Corollary~\ref{cor-decoupling-C-Q-state}. (See also the remark after the proof of the corollary.)
 	
For~$\gamma_2  \in  (0,1)$ chosen as in Corollary~\ref{cor-decoupling-C-Q-state},  
 	let~$a_2 \eqdef |B^{C}C| / \gamma_2$,~$n_2 \eqdef  
 	a_2^{1/ \delta_2^2}$, and~$D_2, D_2'$ and~$E_2$ be quantum registers 
 	with~$|D_2|=|D_2'| \geq n_2$ and~$|E_2|=a_2$. Let
 	\[
 	\nu_2^{B^{C}C E_2} \quad \eqdef \quad \frac{1}{a_2} \sum_{r=1}^{a_2} \density{r}^{B^{C}CE_2} \enspace.
 	\]
 	According to Corollary~\ref{cor-decoupling-C-Q-state},
 	there exists a unitary operator~$U_2 \in \unitary(\cH^{JB^{C}CE_2D_2})$, read-only on register~$J$, and a projection operator~$\Pitilde \in~\Pos(\cH^{JB^{C}CE_2D_2})$ 
 	 such that
 	\begin{multline}
 	\label{eq:Dmax-tau2-decoupled-BcC}
         U_2 \left(  \tau_1^{RB^{R}JB^{C}C} \tensor
         \density{0}^{E_2} \tensor \xi_{a_2:n_2}^{D_2} \right) U_2^{\dagger} 
          \quad
         \leq \quad \log (1+15\delta_2^2) \; \tau_1^{RB^{R}J} \tensor \nu_2^{B^{C}CE_2} \tensor \xi_{1:n_2}^{D_2} \enspace,
    \end{multline}
    \begin{equation}
    \label{eq:decoupled-substate-tau2}
        \Pitilde \br{\tau_1^{R B^{R}J} \tensor \nu_2^{B^{C}CE_2} 
        \tensor \xi_{1:n_2}^{D_2}} \Pitilde  \quad \preceq\quad 2 
        \; U_2 \left( \tau_1^{R B^{R}JB^{C}C} \tensor
         \density{0}^{E_2} \tensor \xi_{a_2:n_2}^{D_2} \right) U_2^{\dagger} \enspace,
    \end{equation}
    and
    \begin{equation}
    \label{eq:trace-proj-tau2}
        \trace \left[ \Pitilde U_2 \left( \tau_1^{RB^{R}JB^{C}C} \tensor
         \density{0}^{E_2} \tensor \xi_{a_2:n_2}^{D_2} \right) U_2^{\dagger} \right] \quad = \quad 1 \enspace.
    \end{equation}
 	Define
 	\begin{equation*}
 	    \tau_2^{RB^{R}JB^{C}CE_2 D_2} 
 	    \quad \coloneqq \quad U_2 \left( \tau_1^{RB^{R}JB^{C}C} \tensor 
         \density{0}^{E_2}\tensor\xi_{a_2:n_2}^{D_2}\right)U_2^{\dagger} \enspace,
 	\end{equation*}
 	and
 	\begin{equation*}
 	    \ket{\kappa_2}^{RAB^{R}JB^{C}CE_2 D_2 D_2'FF'} \quad
 	    \coloneqq \quad U_2 \left( \ket{\kappa_1}^{RAB^{R}JB^{C}C FF'} \tensor 
 	    \ket{0}^{E_2} \tensor \ket{\xi_{a_2:n_2}}^{D_2 D_2'} \right) \enspace.
 	\end{equation*}
 	
 	Since~$U_2$ is read-only on register~$J$ and 
 	~$J$ is classical in the state~$\kappa_1^{RB^{R}JB^{C}C}$, the unitary operator~$U_2$ keeps~$\kappa_1^{RB^{R}J}$ intact.
 	So, we have
 	\begin{equation}
 	\label{eq:kappa2=psi1-RBRJ}
 	    \kappa_2^{RB^{R}J} \quad = \quad \kappa_1^{RB^{R}J} \quad = \quad \psi_1^{RB^{R}J} \enspace.
 	\end{equation}
 	Moreover, by Eq.~\eqref{eq:Dmax-tau2-decoupled-BcC},~$\tau_2$ is close to a product state in max-relative entropy and therefore, we can claim the following statement. 
 	\begin{claim}
 	\label{claim-1}
 	For the state~$\kappa_2$ defined above, we have
 	\begin{align}
 	\label{eq:Dmax-kappa2-product}
 	    \Dmax \mleft( \left. \kappa_2^{RB^{R}JB^{C}CE_2 D_2 F} 
         \right\|~ \kappa_2^{RB^{R}J} \tensor \nu_2^{B^{C}CE_2} \tensor 
         \xi_{1:n_2}^{D_2} \tensor \frac{\id^{F}}{|F|} \mright)
         & \quad \le \quad \Dmax\mleft({\psi}^{RBC} \| ~ \sigma^{RBC} \mright) 
         + 5 \delta_2 \enspace,
    \end{align}
    and
    \begin{align}
    \label{eq:Dh-kappa2-product}
         \Dh^{\epsilon_2^2}\mleft( \left. \kappa_2^{B^{R}JB^{C}CE_2 D_2 F} 
         \right\|~ \kappa_2^{B^{R}J} \tensor \nu_2^{B^{C}CE_2} \tensor 
         \xi_{1:n_2}^{D_2} \tensor \frac{\id^{F}}{|F|} \mright)
         \quad \ge \quad \Dh^{\epsilon_2^4/4}\mleft({\psi}^{BC} \| ~ \sigma^{BC} \mright) - 1 \enspace.
 	\end{align}
 	\end{claim}
We prove the claim at the end of this section.

To redistribute registers~$B^C C$ in the state~$\psi_1$ with the desired
 	cost, Claim~\ref{claim-1} suggests that it would be sufficient for parties 
 	to transform their joint state~$\psi_1$ to~$\kappa_2$ through the unitary
 	operators~$ U_2 U_1$, then use the AJW protocol to redistribute 
 	registers~$B^C C E_2 D_2 F$, and finally, transform back~$\kappa_2$ to 
 	the state~$\psi_1$ by applying~$ U_1^{-1} U_2^{-1}$. However, 
 	in order to apply~$U_2 U_1$, one needs to have access to all the 
 	registers~$J B^C C$, but initially registers~$B^C C$ are with Alice and 
 	register~$J$ is with Bob. This problem can be resolved using the Uhlmann theorem, as in the GHZ example. Recall that~$\kappa_2^{RB^{R}J} = \psi_1^{RB^{R}J}$ as mentioned in Eq.~\eqref{eq:kappa2=psi1-RBRJ}. Therefore, by the Uhlmann Theorem, there exists an isometry~$V: \cH^{AB^{C}C} \rightarrow \cH^{A B^{C} C E_2 D_2 D_2' FF'}$ such that
 	\begin{equation}
 	\label{eq:dist-omega-kappa}
 	    V \ket{\psi_1}^{RAB^{R}JB^{C}C} \quad = \quad \ket{\kappa_2}^{RAB^{R}J B^{C}CE_2D_2D_2'FF'} \enspace.
 	\end{equation}
 	Notice that~$V$ only acts on registers~$A B^C C$ which are initially 
 	with Alice and so she can apply the isometry~$V$ locally  to 
 	transform~$\psi_1$ to~$\kappa_2$.

Now we have all the ingredients for the new state redistribution protocol. We describe the steps systematically below. Let
\begin{equation*}
\beta \quad \coloneqq \quad \Dmax\mleft(\psi^{RBC} \| ~ \sigma^{RBC} \mright) 
         + 5 \delta_2 \enspace,
\end{equation*}
and~$m \coloneqq \ceil{\tfrac{2^{\beta}}{\epsilon_2^{2}}}$, where~$\epsilon_2 \in (0,1)$.
Let~$S$ and~$T$ be quantum registers such that~$|S|=|T|= \abs{B^{C}CE_2D_2 F}$. Let~$\ket{\eta}^{ST}$ be a purification of~$\nu_2^{B^{C}CE_2} \tensor 
         \xi_{1:n_2}^{D_2} \tensor \frac{\id^{F}}{|F|}$ such that~$\eta^{T} = \nu_2^{B^{C}CE_2} \tensor  \xi_{1:n_2}^{D_2} \tensor \frac{\id^{F}}{|F|}$.

\textbf{The protocol.} In order to redistribute~$\ket{\psi}^{RABC}$, Alice and Bob implement the following steps.
\begin{enumerate}
\item Initially, Alice and Bob start in the state~$\ket{\psi}^{RABC}$, and share the quantum state~$\ket{\xi}^{D_1' D_1}$ and~$m$ copies of the state~$\ket{\eta}^{ST}$ in registers~$(S_i T_i : i \in [m])$.
Hence, the initial joint quantum state of Ref, Alice, and Bob is
\[
		\ket{\psi}^{RABC} \tensor \ket{\xi}^{D_1'D_1} \bigotimes_{i=1}^{m} \ket{\eta}^{S_i T_i} \enspace,
\]
such that register~$R$ is held by Ref, registers~$(AC D_1' S_1 \dotsc S_m)$ 
are held by Alice, and registers~$(B D_1 T_1 \dotsc T_m)$ are held by Bob. 

\suppress{
		~$\ket{\xi'}^{X'X}$,
		 $\ket{\xi_{a_2:n_2}}^{D_2'D_2}$ 
		and~$m \coloneqq \ceil{\tfrac{2^{\beta}}{\epsilon_2^{2}}}$ copies of the 
		state~$\ket{\eta}^{ST}$. Hence, the initial joint quantum state of Ref, Alice and Bob is
		\[
		\ket{\psi}^{RABC} \tensor \ket{\xi}^{D_1'D_1} \tensor \ket{\xi_{a_2:n_2}}^{D_2'D_2} \bigotimes_{i=1}^{m} 
		\ket{\eta}^{S_i T_i} \enspace,
		\]
		such that register~$R$ is held by Ref, registers~$(ACS_1\ldots S_m D_1'D_2')$ 
		are held by Alice and registers~$(BT_1 \ldots T_m D_1 D_2)$ are held by Bob.
}

		\item Alice and Bob pre-process their joint state via local transformations, without any communication. I.e., Bob
		applies the isometry~$U_{\rii,B} U_{\ri}$ on his registers, and Alice applies 
		the isometry~$U_{\rii,A}$ on her registers. This transforms their joint state on~$RABC D_1' D_1$  
		into a quantum state~$\psi_2^{R A B^R J B^C C D_1' D_1}$ which has purified distance at most~$\delta_1$ from~$\psi_1^{R A B^R J B^C C} \tensor \xi^{D_1' D_1} $, where the state~$\psi_1$ is as given by Eq.~\eqref{eq-psi1}. 
		
		At this point, the registers~$(A B^{C} C D_1')$ are with Alice, registers~$(B^{R} J D_1)$ are with Bob, and register~$(R)$ is with Ref. Registers~$(S_i T_i)$ are not touched in this step, and are shared as before. Registers~$D_1' D_1$ are not used after this point, and may be discarded.

\item
Alice and Bob perform the first part of reduction involving the coherent measurement and the decoupling of a classical-quantum state.
I.e., Alice applies the isometry~$V$ to the registers~$A B^C C$.
This transforms their joint state on registers~$R A B^R J B^C C$
into a quantum state~$\omega$ which has purified distance at most~$\delta_1$  from~$\ket{\kappa_2}^{R A B^{R} J B^{C} C E_2 D_2 D_2' F F'}$.

The registers~$(AB^{C}C E_2 D_2 D_2' FF')$ are with Alice, registers~$(B^{R} J)$ are with Bob, and register~$(R)$ is with Ref. Registers~$(S_i T_i)$ are not touched in this step, and are shared as before.
	
\item 
\label{step-ajw}
Alice and Bob run the AJW protocol to transfer the registers~$B^{C}C E_2 D_2 F$ to Bob, as described in Section~\ref{sec-ajw}. 
I.e., the two parties redistribute their registers assuming that their joint state is~$\ket{\kappa_2}^{RAB^{R}JB^{C}CE_2D_2D_2'FF'}$, with the registers held as above. For this, they use the~$m$ copies of the state~$\ket{\eta}^{ST}$ that were shared in registers~$(S_i T_i : i \in [m])$.

For the reader's convenience we include in Table~\ref{tab:ajw-corresp} the correspondence between the states and registers involved in the AJW protocol as presented in Section~\ref{sec-ajw} and those involved in the use of the protocol here.

\begin{table}[t]
    \centering
    \begin{tabular}{|l|c|c|}
    \hline
    & Section~\ref{sec-ajw} & Here \\
    \hline
    & & \\
    State to be redistributed (``input'') & $ \ket{\psi}^{RABC}$ & $ \ket{\kappa_2}^{RAB^{R}JB^{C}CE_2D_2D_2'FF' }$  \\
    & & \\
    Registers of input initially with Ref & $R$ & $R$ \\
    & & \\
    Registers of input initially with Alice & $A$ & $AB^{C}C E_2 D_2 D_2' FF' $ \\
    & & \\
    Registers of input initially with Bob & $B $ & $B^{R} J$ \\
    & & \\
    Registers to be transferred to Bob & $C$ & $B^{C}C E_2 D_2 F$ \\
    & & \\
    Smoothed state & $ \psi'^{RBC}$ & $\kappa_2^{R B^{R} J B^{C}C E_2 D_2 F}$ \\
    & & \\
    State used in application of Convex-Split & $\sigma^C$ & $\nu_2^{B^{C}CE_2} \tensor  \xi_{1:n_2}^{D_2} \tensor \frac{\id^{F}}{|F|}$ \\
    & & \\
    Initial shared entangled state & $ \bigotimes_{i = 1}^m \ket{\sigma}^{L_i C_i}$ & $ \bigotimes_{i = 1}^m \ket{\eta}^{S_i T_i} $ \\
    & & \\
    Registers of entangled state initially with Alice & $L_1 \dotsb L_m$ & $S_1 \dotsb S_m $ \\
    & & \\
    Registers of entangled state initially with Bob & $C_1 \dotsb C_m$ & $T_1 \dotsb T_m $ \\
    & & \\
    \hline
    \end{tabular}
    \caption{The correspondence between the states and registers in the AJW protocol as described in Section~\ref{sec-ajw} and those involved in the use of the AJW protocol here.}
    \label{tab:ajw-corresp}
\end{table}

At the end of the AJW protocol, the parties end up with a state~$\widehat{\omega}^{RAB^{R}JB^{C}CE_2D_2D_2'FF'}$ such 
that register~$(R)$ is held  with Ref,~$(AD_2'F')$ are held with Alice and~$(B^{R}JB^{C}CE_2D_2F)$ are held with Bob.
	
		\item
		\label{step-2nd-part}
		Bob completes the second part of reduction involving the coherent measurement and the decoupling of a classical-quantum state and reverses the first pre-processing step. I.e., he applies the operator~$(U_2U_1U_{\ri})^{-1}$ on registers~$B^{R}JB^{C}CE_2D_2F$.
		
		\item The output of the protocol is now the state in registers~$RABC$.
	\end{enumerate}
	
According to Theorem~\ref{thm-QSR-AJW}, the communication cost of this protocol is
	\begin{IEEEeqnarray*}{l}
	   \frac{1}{2} \left[ \Dmax \mleft( \left. \kappa_2^{RB^{R}JB^{C}CE_2 D_2 F} 
         \right\|~ \kappa_2^{RB^{R}J} \tensor \nu_2^{B^{C}CE_2} \tensor 
         \xi_{1:n_2}^{D_2} \tensor \frac{\id^{F}}{|F|} \mright) \right.\\ \qquad
         - \> \left. \Dh^{\epsilon_2^{2}}\mleft( \left. \kappa_2^{B^{R}JB^{C}CE_2 D_2 F} 
         \right\|~ \kappa_2^{B^{R}J} \tensor \nu_2^{B^{C}CE_2} \tensor 
         \xi_{1:n_2}^{D_2} \tensor \frac{\id^{F}}{|F|} \mright) \right]
         + \log \frac{1}{\epsilon_2^{2}} \\
    \end{IEEEeqnarray*}
    which is at most
    \begin{IEEEeqnarray*}{l}
         \quad \frac{1}{2} \left[\Dmax\mleft({\psi}^{RBC} \| ~ \sigma^{RBC} \mright) - \Dh^{\epsilon_2^4/4} \mleft({\psi}^{BC} \| ~ \sigma^{BC} \mright) \right]
         + 5 \, \delta_2 + \log \frac{1}{\epsilon_2^{2}} + 1 \enspace, 
	\end{IEEEeqnarray*}
	by Claim~\ref{claim-1}.
	
\textbf{Correctness of the protocol.} Let~$\phi$ be the final joint state of parties in the above protocol. We have
 	 \begin{IEEEeqnarray*}{rCl}
 	 \IEEEeqnarraymulticol{3}{l}{
 	 	\rP \left( \phi^{RAB C}, \; \psi^{RABC} \right) } \\ 
 	 	& \leq & \quad \rP\left( \phi^{RA B C E_2D_2 D_2'FF'},  ~ {\psi}^{RABC} \tensor \density{0}^{E_2} \tensor\xi_{a_2:n_2}^{D_2D_2'} \tensor \Psi^{FF'} \right) \\
 	 	& \le & \quad \rP\left( \widehat{\omega}^{RAB^{R}JB^{C}CE_2D_2D_2'FF'}, ~ \kappa_2^{RAB^{R}JB^{C}CE_2D_2D_2'FF'} \right) \\
 	 	& \leq &\quad \rP \left(\widehat{\omega}^{RAB^{R}JB^{C}CE_2D_2D_2'FF'}, ~ \omega^{RAB^{R}JB^{C}CE_2D_2D_2'FF'} \right) \\
 	 	&& \quad + \> \rP \left(\omega^{RAB^{R}J B^{C}CE_2D_2D_2'FF'}, ~ \kappa_2^{RAB^{R}JB^{C}CE_2D_2D_2'FF'} \right) \\
 	 	& \leq & \quad   9 \epsilon_2 + \delta_1  \enspace.
 	 \end{IEEEeqnarray*}
Here, the first and second inequalities follow from monotonicity of purified distance under quantum operations. In the first, we consider the extensions of the two states to a larger set of registers. In the second inequality, we consider the states by reversing the isometries in step~\ref{step-2nd-part} of the protocol. The third inequality is the Triangle Inequality for purified distance. The last inequality holds since~$\widehat{\omega} \in \sB^{9 \epsilon_2}(\omega)$ by Theorem~\ref{thm-QSR-AJW}, and~$\omega \in \sB^{\delta_1}(\kappa_2)$. 
 	 	
By the properties of the embezzlement protocol due to van Dam and Hayden~\cite{vDH03-embezzling} (see Eqs.~\eqref{eq:vDH-embezzlement} and~\eqref{eq:vDH-embezzlement-error}) and the protocol given by Corollary~\ref{cor-decoupling-C-Q-state}, we can make~$\delta_1$ 
and~$\delta_2$ arbitrarily small by choosing suitable
entangled states shared between Alice and Bob. (Note that this comes at the cost of shared entanglement with arbitrarily large local dimension.)
Hence, the statement of the theorem follows.

It only remains to prove Claim~\ref{claim-1}.

\begin{proofof}{Claim~\ref{claim-1}}
	Consider the states and operators defined in the description preceding the protocol.  Since register~$J$ is classical in both~$\kappa_1^{RB^{R}JB^{C}C}$ 
 	and~$\tau_1^{RB^{R}JB^{C}C}$ and~$U_2$ is read-only on~$J$, we have 
 	that~$\kappa_2^{RB^{R}J} = \tau_2^{RB^{R}J} = \tau_1^{RB^{R}J}$.
 	Therefore, we get
 	\begin{IEEEeqnarray*}{rCl}
         \IEEEeqnarraymulticol{3}{l}
 	     {\Dmax \mleft( \left. \kappa_2^{RB^{R}JB^{C}CE_2 D_2 F} 
         \right\|~ \kappa_2^{RB^{R}J} \tensor \nu_2^{B^{C}CE_2} \tensor 
         \xi_{1:n_2}^{D_2} \tensor \frac{\id^{F}}{|F|} \mright)} \\
         & \leq & \quad \Dmax \mleft( \left. \kappa_2^{RB^{R}JB^{C}CE_2 D_2 F}
         \right\|~ \tau_2^{RB^{R}JB^{C}CE_2 D_2}\tensor \frac{\id^{F}}{|F|}\mright) \\ \quad
         & & \quad + \> \Dmax \mleft( \left. \tau_2^{RB^{R}JB^{C}CE_2 D_2}
         \right\|~ \tau_2^{RB^{R}J} \tensor \nu_2^{B^{C}CE_2} \tensor 
         \xi_{1:n_2}^{D_2} \mright) \\ \quad
         & \le & \quad \Dmax\mleft({\psi}^{RBC} \| ~ \sigma^{RBC} \mright) 
         + \log (1+15\delta_2^2)  \enspace,
 	\end{IEEEeqnarray*}
 	where the last inequality is a consequence of 
 	Eq.~\eqref{eq:Dmax-tau2-decoupled-BcC} and the fact 
 	that~$\kappa_2^{RB^{R}JB^{C}CE_2 D_2 F}$ and $\tau_2^{RB^{R}JB^{C}CE_2 D_2 F}$
 	are obtained by the applying the same unitary transformation to~${\psi}^{RBC}$ and~$\sigma^{RBC}$, respectively.
 	The above equation implies Eq.~\eqref{eq:Dmax-kappa2-product} since~$\log_2(1+15x^2) \le 5x$ for all~$x \geq 0$.
 	
 	In the rest of the proof, we show that
  	\begin{IEEEeqnarray}{rCl}
 	    \IEEEeqnarraymulticol{3}{l}
         {\Dh^{\epsilon_2^2}\mleft( \left. \kappa_2^{B^{R}JB^{C}CE_2 D_2 F} 
         \right\|~\kappa_2^{B^{R}J} \tensor \nu_2^{B^{C}CE_2} \tensor 
         \xi_{1:n_2}^{D_2} \tensor \frac{\id^{F}}{|F|} \mright)} \nonumber \\ \quad 
         & \ge & \quad \Dh^{\epsilon_2^4/4}\mleft( \left. \kappa_2^{B^{R}JB^{C}CE_2 D_2 F}  \right\|~\tau_2^{B^{R}JB^{C}CE_2 D_2} \tensor \frac{\id^{F}}{|F|} \mright) - 1  \enspace. \label{eq:Dh-middle-step}
 	\end{IEEEeqnarray}
 	Then, Eq.~\eqref{eq:Dh-kappa2-product} follows  since~$\kappa_2^{RB^{R}JB^{C}CE_2 D_2 F}$ and $\tau_2^{RB^{R}JB^{C}CE_2 D_2 F}$
 	are obtained by the applying the same unitary transformation to~${\psi}^{RBC}$ and~$\sigma^{RBC}$, respectively.
 	Let
\[
 	\lambda \quad \coloneqq \quad \Dh^{\epsilon_2^4/4}\mleft( \left. \kappa_2^{B^{R}JB^{C}CE_2 D_2 F}  \right\|~\tau_2^{B^{R}JB^{C}CE_2 D_2F} \mright) \enspace, 
\]
 	and~$\Pi'$ be the POVM operator achieving~$\lambda$, i.e.,
 	\begin{IEEEeqnarray*}{rCl}
 	\trace \mleft[ \Pi' \kappa_2^{B^{R}JB^{C}CE_2 D_2 F} \mright] \quad \ge \quad 1- \frac{\epsilon_2^4}{4}
 	\end{IEEEeqnarray*}
 	and
 	\begin{equation*}
 	    \trace \mleft[ \Pi' \left(\tau_2^{B^{R}JB^{C}CE_2 D_2} \tensor \frac{\id^{F}}{|F|} \right) \mright] \quad = \quad 2^{-\lambda} \enspace.
 	\end{equation*}
     Recall that~$\kappa_2^{B^R J} = \tau_2^{B^R J} =\tau_1^{B^R J}$. So, Eq.~\eqref{eq:decoupled-substate-tau2} implies that
    \begin{equation}
    \label{eq:P-kappa2-product}
         \Pitilde \left( \kappa_2^{B^{R}J} \tensor \nu_2^{B^{C}CE_2} \tensor 
         \xi_{1:n_2}^{D_2} \right) \Pitilde \quad \preceq\quad  2 \; \tau_2^{B^{R}JB^{C}CE_2D_2} \enspace.
    \end{equation}
    Since~$\sigma^{RBC} \in \ME^{\epsilon_2^4/4, \psi}$, the state~$\kappa_2^{JB^{C}CE_2D_2}$ is~$(\epsilon_2^4/4)$-close to~$\tau_2^{JB^{C}CE_2D_2}$ in purified distance. This implies that  
    \begin{equation}
    \label{eq:Tr-P-kappa2}
        \trace \mleft[ \Pitilde \, \kappa_2^{B^{R}JB^{C}CE_2D_2F}  \mright] \quad \ge \quad \trace \mleft[ \Pitilde \, \tau_2^{B^{R}JB^{C}CE_2D_2F}  \mright] - \frac{\epsilon_2^4}{4} \quad = \quad 1 - \frac{\epsilon_2^4}{4} \enspace,
    \end{equation}
    using Theorem~\ref{prop-trnorm-povm}, Theorem~\ref{prop-FvdG}, and Eq.~\eqref{eq:trace-proj-tau2}. So, the Gentle Measurement lemma, Lemma~\ref{prop-gentl-meas}, implies that
    \begin{equation}
    \label{eq-trnorm-kappa_2-Pitilde}
        \trnorm{\frac{\Pitilde \, \kappa_2^{B^{R}JB^{C}CE_2D_2F} \Pitilde}{\trace \mleft[ \Pitilde \kappa_2^{B^{R}JB^{C}CE_2D_2F}  \mright]} - \kappa_2^{B^{R}JB^{C}CE_2D_2F} } \quad \le \quad \epsilon_2^2 \enspace.
    \end{equation}
    Define the POVM operator~$\Pi \coloneqq \Pitilde \, \Pi' \Pitilde$. By Eq.~\eqref{eq-trnorm-kappa_2-Pitilde}, Eq.~\eqref{eq:Tr-P-kappa2}, and Theorem~\ref{prop-trnorm-povm} we have
    \begin{IEEEeqnarray*}{rCl's}
        \trace \mleft[ \Pi \, \kappa_2^{B^{R}JB^{C}CE_2D_2F} \mright] \quad 
        & = & \quad  \trace \mleft[ \Pi' \Pitilde \, \kappa_2^{B^{R}JB^{C}CE_2D_2F} \Pitilde \mright] \\
        & \geq & \quad \left(1 - \frac{\epsilon_2^4}{4}\right) \left(\trace \mleft[ \Pi' \kappa_2^{B^{R}JB^{C}CE_2D_2F} \mright] - \frac{\epsilon_2^2}{2} \right) 
        & \\  
        & \geq & \quad 1-\epsilon_2^2 \enspace.
    \end{IEEEeqnarray*}
By Eq.~\eqref{eq:P-kappa2-product}, we get
    \begin{IEEEeqnarray*}{rCl}
     \trace \mleft[ \Pi \left( \kappa_2^{B^{R}J} \tensor \nu_2^{B^{C}CE_2} \tensor 
         \xi_{1:n_2}^{D_2} \tensor \frac{\id^{F}}{|F|} \right) \mright] \quad 
         & \leq & \quad 2  \trace \mleft[ \Pi' \left( \tau_2^{B^{R}JB^{C}CE_2D_2} \tensor \frac{\id^{F}}{|F|} \right) \mright] \\
         & = & \quad 2^{-\lambda +1} \enspace,
    \end{IEEEeqnarray*}
which implies Eq.~\eqref{eq:Dh-middle-step}, as desired.
\end{proofof}

\subsection{Asymptotic and i.i.d.\ analysis}
\label{sec-aiid-qsr}

We can obtain the asymptotic cost of redistributing copies of a state using the one-shot bound from the previous section. 
Suppose that the state~$\ket{\psi}^{R^n A^n B^n C^n} \coloneqq \Big( \ket{\psi}^{RABC} \Big)^{\tensor n}$ is shared between Alice~$(A^nC^n)$, Bob~$(B^n)$ and Ref~$(R^n)$ where~$R^{n}$,~$A^n$,~$B^{n}$, and~$C^{n}$ denote~$n$-fold tensor products of registers~$R$,~$A$,~$B$ and~$C$, respectively. Let~$\epsilon \coloneqq \epsilon_1 = \epsilon_2^4/4$. By Theorem~\ref{thm-main}, choosing~$\sigma^{R^n B^n C^n} \eqdef \psi'^{R^n B^n} \tensor \psi^{C^n}$, there exists an entanglement-assisted one-way protocol which outputs a state~$\phi^{R^n A^n B^n C^n} \in \sB^{14\epsilon^{1/4}} (\psi^{R^n A^n B^n C^n})$ with communication cost~$Q(n,\epsilon)$ bounded as
\begin{IEEEeqnarray*}{rCl}
    \IEEEeqnarraymulticol{3}{l}{Q(n, \epsilon)} \\
     &\le& \quad \frac{1}{2}  \inf_{\psi' \in \sB^{\epsilon}(\psi^{R^n B^n C^n})}  
     \left[\Dmax\mleft({\psi'}^{R^n B^n C^n} 
    \Big\| ~ \psi'^{R^n B^n} \tensor \psi^{C^n} \mright) - \Dh^{\epsilon} 
    \mleft({\psi'}^{B^n C^n} \Big\| ~ \psi'^{B^n} \tensor \psi^{C^n} \mright) \right]
    + \log \frac{1}{2 \sqrt{\epsilon}} \\
     & \le & \quad \frac{1}{2}  \inf_{\substack{\psi' \in \sB^{\epsilon}(\psi^{R^n B^n C^n}) \\ \psi'^{R^n B^n} = \psi^{R^n B^n}}}  
     \left[\Dmax\mleft({\psi'}^{R^n B^n C^n} 
    \Big\| ~ \psi^{R^n B^n} \tensor \psi^{C^n} \mright) - \Dh^{\epsilon} 
    \mleft({\psi'}^{B^n C^n} \Big\| ~ \psi^{B^n} \tensor \psi^{C^n} \mright) \right]
    + \log \frac{1}{2 \sqrt{\epsilon}} \\
     & \le & \quad \frac{1}{2}  \inf_{\substack{\psi' \in \sB^{\epsilon}(\psi^{R^n B^n C^n}) \\ \psi'^{R^n B^n} = \psi^{R^n B^n}}}  
     \left[\Dmax\mleft({\psi'}^{R^n B^n C^n} 
    \Big\| ~ \psi^{R^n B^n} \tensor \psi^{C^n} \mright)  - \Dh^{2\epsilon} 
    \mleft({\psi}^{B^nC^n} \Big\| ~ \psi^{B^n} \tensor \psi^{C^n} \mright) \right]
    + \log \frac{1}{2 \sqrt{\epsilon}} \\
     & \le & \quad  \frac{1}{2}    
     \left[\Dmax^{\epsilon/3}\mleft({\psi}^{R^n B^n C^n} 
    \Big\| ~ \psi^{R^n B^n} \tensor \psi^{C^n} \mright)  - \Dh^{2\epsilon} 
    \mleft({\psi}^{B^nC^n} \Big\| ~ \psi^{B^n} \tensor \psi^{C^n} \mright) \right]
    + \log \frac{1}{2 \sqrt{\epsilon}} + \log \frac{72+\epsilon^2}{\epsilon^2}\enspace,
\end{IEEEeqnarray*}
where the first inequality follows from Eq.\eqref{eq:main-cost}, the third inequality follows from the definition of Hypothesis testing entropy, and the last inequality follows from Theorem~\ref{prop-Dmax-psmooth} for the choice of ~$\epsilon,\delta \leftarrow \epsilon/3$,~$\rho^{AB} \leftarrow \psi^{R^n B^n C^n}$,~$\rho^{A} \leftarrow \psi^{R^n B^n}$ and~$\sigma^{B} \leftarrow \psi^{C^n}$. Therefore, using Theorem~\ref{prop-QAEP}, the asymptotic communication rate of redistributing~$n$ copies of a pure state~$\ket{\psi}^{RABC}$ is
\begin{equation*}
    \lim_{n \rightarrow \infty} \; \frac{1}{n} \, Q(n,\epsilon) \quad \le \quad 
    \frac{1}{2} \, \rI(R:C \,|\, B)_{\psi} \enspace.
\end{equation*}

\section{Conclusion and outlook}
\label{sec-outlook}

In this article, we revisited the task of one-shot quantum state redistribution, and introduced a new protocol achieving this task with communication cost
\begin{equation}
    \frac{1}{2} \min_{\psi' \in \sB^{\epsilon}(\psi^{RBC})} \min_{\sigma^{RBC}\in \ME^{\epsilon^2/4,\psi'}}\mleft[\Dmax \br{\psi'^{RBC}\|\sigma^{RBC}}-\Dh^{\epsilon}\br{\psi'^{BC}\|\sigma^{BC}}\mright]+ \Order \! \Big( \log\frac{1}{\epsilon} \Big) \enspace,
\end{equation}
with error parameter~$\epsilon$. This is the first result connecting the communication cost of state redistribution with Markov chains. It provides an operational interpretation for a one-shot representation of quantum conditional mutual information as explained in Sec~\ref{sec:QSR-Intro}. In the special case where $\psi^{RBC}$ is a quantum Markov chain, our protocol leads to near-zero communication which was not known for the previous protocols designed for arbitrary states. Moreover, the communication cost of our protocol is lower than that of all previously known one-shot protocols and we show that it achieves the optimal cost of~$ \tfrac{1}{2} \, \rI(R:C \, | \, B)$ in the asymptotic i.i.d.\ setting. Our protocol also achieves the near-optimal result of ref.~\cite{AJW17-classical-compression} in the case when~$\psi^{RBC}$ is classical.  

A question of interest is whether the communication cost of our one-shot protocol can be bounded with~$\rI(R : C \, | \, B)$. In the quantum communication complexity setting, such a bound would imply the possibility of compressing the communication of bounded-round quantum protocols to their information content. This would lead to a \emph{direct-sum theorem\/} for bounded-round quantum communication complexity~\cite{Touchette15-QIC}.

Another question that we have not addressed in this article is whether our bound is near-optimal. There are several known lower bounds in the literature for the communication cost of entanglement-assisted quantum state redistribution, such as in ref.~\cite[Proposition~6]{BCT16-state-redistribution} and ref.~\cite[Theorem~3.2, Eq.~(3.17)]{LWD16-converse-QSR}. However, it is not clear if our bound matches any of them. Obtaining a near-optimal bound for one-shot quantum state redistribution remains a major open question.

\bibliographystyle{plain}
\bibliography{bibl}

\end{document}